\documentclass[preprint]{elsarticle}
\journal{arXiv}
\usepackage{times}
\usepackage{helvet}
\usepackage{courier}
\usepackage{makecell}
\usepackage[small]{caption}

\usepackage{amsthm,amsmath,amssymb,bbm,graphicx,rotating}
\usepackage{boxedminipage}
\usepackage{multirow}
\usepackage{enumerate}
\usepackage{epic,eepic}
\usepackage{tikz}
\usetikzlibrary{shapes}
\usepackage[rightcaption]{sidecap}
\usepackage{booktabs}
\usepackage{graphicx}
\usepackage{paralist}
\usepackage{microtype}

\usepackage{floatrow}

\usepackage{fullpage}

\tikzstyle{every picture} = [>=latex]
\usepackage{xspace}
\newcommand{\bigoh}{\mathcal{O}}

\renewcommand{\phi}{\varphi}

\makeatletter
\newcommand*{\ol}[1]{$\overline{\hbox{#1}}\m@th$}

\makeatother

\newcommand{\Nat}{{\mathbb{N}}}
\newcommand{\cA}{{\cal A}}

\newcommand{\cF}{{\cal F}}

\newcommand{\var}{\text{\normalfont var}}

\newcommand{\III}{\mathbf{I}}
\newcommand{\SSS}{\mathcal{S}}

\newcommand{\SB}{\left\{\,}%
\newcommand{\SM}{\;{:}\;}%
\newcommand{\SE}{\,\right\}}%

\newcommand{\fr}{\mathfrak{p}}
\newcommand{\ms}{m_{\text{S}}}
\newcommand{\mm}{m_{\text{M}}}
\newcommand{\mb}{m_{\text{L}}}

\newcommand{\eval}{\eta}

\theoremstyle{plain}
\newtheorem{THE}{Theorem}
\newtheorem{PRO}[THE]{Proposition}
\newtheorem{LEM}[THE]{Lemma}
\newtheorem{COR}[THE]{Corollary}
\newtheorem{theorem}[THE]{Theorem}

\newtheorem{lemma}[THE]{Lemma}

\newtheorem*{THE*}{Theorem}
\newtheorem*{fact*}{Fact}
\theoremstyle{definition}
\newtheorem{DEF}[THE]{Definition}
\newtheorem*{DEF*}{Definition}
\theoremstyle{remark}

\newtheorem{CLM}{Claim}

\newcommand{\no}{\textsc{No}}

\newcommand{\pbDefPtask}[4]{%
\noindent
\begin{center}
\begin{boxedminipage}{0.98 \columnwidth}
#1\\[5pt]
\begin{tabular}{l p{0.70 \columnwidth}}
Input: & #2\\
Parameter: & #3\\
Task: & #4
\end{tabular}
\end{boxedminipage}
\end{center}
}

\newcommand{\cc}[1]{{\mbox{\textnormal{\textsf{#1}}}}\xspace}  

\newcommand{\NP}{\cc{NP}}

\newcommand{\coNPpoly}{\cc{co-NP/poly}}
\newcommand{\FPT}{\cc{FPT}}
\newcommand{\XP}{\cc{XP}}
\newcommand{\Weft}{{\cc{W}}}
\newcommand{\W}[1]{{\Weft}{{[#1]}}}
\newcommand{\paraNP}{\cc{pNP}}

\newcommand{\COMPFRAC}[1]{\textsc{Backdoor Detection}}
\newcommand{\VDTSC}[1]{\textsc{Fracture Vertex Deletion}}
\newcommand{\COMPFRACS}[1]{\textsc{BD}}
\newcommand{\VCOMPFRACS}[1]{\textsc{V\hy BD}}
\newcommand{\CCOMPFRACS}[1]{\textsc{C\hy BD}}
\newcommand{\VDTSCS}[1]{\textsc{FVD}}

\newcommand{\TSAT}{\textsc{$3$\hy Satisfiability}}

\newcommand{\problembox}[4]{
  \begin{center}
    \fbox{
      \parbox{0.95\columnwidth}{
        #1\\[0.3em]
        \renewcommand{\tabcolsep}{3pt}
        \begin{tabular}{rp{0.70\columnwidth}}
          \textit{Instance:\ } & #2\\
          \textit{Parameter:\ } & #3\\
          \textit{Question:\ } & #4
        \end{tabular}
      }
    }
  \end{center}
}

\newcommand{\tuple}[1]{\langle{#1}\rangle}  

\newcommand{\hy}{\hbox{-}\nobreak\hskip0pt}

\newcommand{\nn}{\mathbb{N}}

\newcommand{\bigO}[1]{\ensuremath{{\mathcal O}(#1)}}

\newcommand{\probfont}[1]{\textnormal{\textsc{#1}}}

\newcommand{\mymax}{\texttt{max}}
\newcommand{\bmax}[1]{B_\mymax}
\newcommand{\Amax}[1]{A_\mymax}
\newcommand{\xmax}[1]{X_\mymax}

\newcommand{\pcmax}[1]{\Gamma}
\newcommand{\twi}[1]{\textup{tw}_I}

\newcommand{\ve}[1]{\mathchoice{\mbox{\boldmath$\displaystyle\bf#1$}}
{\mbox{\boldmath$\textstyle\bf#1$}}
{\mbox{\boldmath$\scriptstyle\bf#1$}}
{\mbox{\boldmath$\scriptscriptstyle\bf#1$}}}

\newcommand{\vea}{{\ve a}}
\newcommand{\veb}{{\ve b}}
\newcommand{\ved}{{\ve d}}

\newcommand{\ver}{{\ve l}}
\newcommand{\veu}{{\ve u}}
\newcommand{\vev}{{\ve v}}
\newcommand{\vew}{{\ve w}}
\newcommand{\vex}{{\ve x}}
\newcommand{\vey}{{\ve y}}
\newcommand{\ves}{{\ve s}}

\newcommand{\coA}{c_{\mathbf{A}}}
\newcommand{\coAi}{c_{\mathbf{A}_{(*,\{i\})}}}
\newcommand{\coQ}{c_{\mathbf{Q}}}
\newcommand{\coB}{c_{\veb}}
\newcommand{\coD}{c_{\ved}}

\begin{document}
\begin{frontmatter}
  \title{The Complexity Landscape of Decompositional Parameters for ILP III.: \\
  Programs with Few Global Variables and Constraints}

  \cortext[cor1]{Corresponding authors}

  \author[uk]{Pavel Dvo\v{r}\'{a}k}\ead{koblich@iuuk.mff.cuni.cz}
    \author[rh]{Eduard Eiben}\ead{Eduard.Eiben@rhul.ac.uk}

    \author[ac]{Robert Ganian\corref{cor1}}
  \ead{rganian@gmail.com}
      \author[dk]{Du\v{s}an Knop}\ead{dusan.knop@fit.cvut.cz }
  \author[she]{Sebastian Ordyniak}
  \ead{sordyniak@gmail.com}

  \address[uk]{Computer Science Institute, Charles University, Prague, Czech Republic}
  \address[rh]{Royal Holloway, University of London, United Kingdom}
  \address[ac]{Algorithms and Complexity Group, Vienna University of Technology, Austria}
  \address[dk]{Czech Technical University in Prague, Czech Republic}
  \address[she]{Department of Computer Science, University of Sheffield, United Kingdom}



\begin{abstract}
Integer Linear Programming (ILP) has a broad range of applications in various areas of artificial intelligence.
Yet in spite of recent advances, we still lack a thorough understanding of which structural restrictions make ILP tractable. Here we study ILP instances consisting of a small number of ``global'' variables and/or constraints such that the remaining part of the instance consists of small and otherwise independent components; this is captured in terms of a structural measure we call \emph{fracture backdoors} which generalizes, for instance, the well-studied class of $N$-fold ILP instances.

Our main contributions can be divided into three parts. First, we formally develop fracture backdoors and obtain exact and approximation algorithms for computing these. Second, we exploit these backdoors to develop several new parameterized algorithms for ILP; the performance of these algorithms will naturally scale based on the number of global variables or constraints in the instance. Finally, we complement the developed algorithms with matching lower bounds. Altogether, our results paint a near-complete complexity landscape of ILP with respect to fracture backdoors.
\end{abstract}


\end{frontmatter}

\section{Introduction}
\label{sec:intro}

Integer Linear Programming (ILP) is the archetypical representative of an \NP{}\hy complete optimization problem and has a broad range of applications in various areas of artificial intelligence. In
particular, a wide variety of problems in artificial intelligence are efficiently solved in practice via a translation into an ILP, including problems from areas such as planning~\cite{BrielVossenKambhampati05,VossenBallLotemNau99}, process scheduling~\cite{FloudasLin05}, packing~\cite{LodiMM02}, vehicle routing~\cite{Toth01}, and network hub location~\cite{AlumurK08}.

In spite of recent advances~\cite{GanianO18,GanianOrdyniakRamanujan17,JansenKratsch15esa}, we still lack a deep understanding of which structural restrictions make ILP tractable. The goal of this line of research is to identify structural properties (formally captured by a numerical \emph{structural parameter} $k$)
which allow us to solve ILP efficiently. In particular, one seeks to either solve an ILP instance $\III$ in time $f(k)\cdot |\III|^{\bigoh(1)}$ (a so-called \emph{fixed-parameter algorithm}), or at least in time $|\III|^{f(k)}$ (a so-called \emph{XP algorithm}), where $f$ is a computable function. This approach lies at the core of the now well-established \emph{parameterized complexity} paradigm~\cite{DowneyFellows13,CyganFKLMPPS15} and has yielded deep results capturing the tractability and intractability of numerous prominent problems in diverse areas of computer science---such as Constraint Satisfaction, SAT, and a plethora of problems on directed and undirected graphs.

In general, structural parameters can be divided into two groups based on the way they are designed. \emph{Decompositional parameters} capture the structure of instances by abstract tools called decompositions; treewidth is undoubtedly the most prominent example of such a parameter, and previous work has obtained a detailed complexity map of ILP with respect to the treewidth of natural graph representations of instances~\cite{GanianO18,GanianOrdyniakRamanujan17}. On the other hand, \emph{backdoors} directly measure the ``distance to triviality'' of an instance: the number of simple operations required to put the instance into a well-defined, polynomially tractable class. While the backdoor approach has led to highly interesting results for problems such as Constraint Satisfaction~\cite{GaspersMOSZ17} and SAT~\cite{GaspersSzeider12}, it has so far been left mostly unexplored in the arena of ILP.

\subsection{Our Contribution}
Here, we initiate the study of backdoors to triviality for ILP by analyzing backdoors which fracture the instance into small, easy-to-handle components. Such \emph{fracture backdoors} can equivalently be viewed as measuring the number of global variables or global constraints in an otherwise ``compact'' instance; in fact, we identify and analyze three separate cases depending on whether we allow global variables only, global constraints only, or both.
We obtain a near-complete complexity landscape for the considered parameters: in particular, we identify the circumstances under which they can be used to obtain fixed-parameter and \XP{} algorithms for ILP, and otherwise prove that such algorithms would violate well-established complexity assumptions. Our results are summarized in the following Table~\ref{fig:results} (formal definitions are given in Section~\ref{sec:fracture}).

\begin{table}[ht]
  \begin{tabular}{lccc}
    \toprule
\hspace{-0.15cm} & Variable & Constraint & Mixed \\ \midrule 
\hspace{-0.15cm}param. & \FPT{} (Cor.~\ref{cor:parvarconFPT}) & \FPT{} (Cor.~\ref{cor:parvarconFPT}) & \XP{} (Cor.~\ref{cor:parmixedXP})\\
\hspace{-0.15cm}unary & \paraNP{}\hy c & \XP{}, \W{1}\hy h
                                         & \paraNP{}\hy c \\
\hspace{-0.15cm} & (Th~\ref{thm:unaryvarhard}) &
                                                   (Th~\ref{thm:unaryconXP},~\ref{thm:unaryconhard}) & (Th~\ref{thm:unaryvarhard}) \\

\hspace{-0.15cm}arbitrary & \paraNP{}\hy c & \paraNP{}\hy c (Th~\ref{lem:binaryconstbdhard})
                                         & \paraNP{}\hy c \\

\end{tabular}
\caption{Complexity landscape for fracture backdoors. Columns distinguish whether we consider variable backdoors, constraint backdoors, or mixed backdoors. Rows correspond to restrictions placed on coefficients in the ILP instance. \vspace{-0.3cm}}
\label{fig:results}
\end{table}

As is evident from the table, backdoor size on its own is not sufficient to break the NP-hardness of ILP; this is far from surprising, and the same situation arose in previous work on treewidth. However, while positive results on treewidth (as well as other considered decompositional parameters such as \emph{torso-width}~\cite{GanianOrdyniakRamanujan17}) required the imposition of domain restrictions on variables, in the case of backdoors one can also deal with instances with unrestricted variable domains---by instead restricting the values of coefficients which appear in the ILP instance. Here, we distinguish three separate cases (corresponding to three rows in Table~\ref{fig:results}): coefficients bounded by the parameter value, coefficients which are encoded in unary, and no restrictions. It is worth noting that in the case of treewidth, ILP remains NP-hard even when coefficients are restricted to $\pm 1$ and $0$.

Our results in row $1$ represent a direct generalization of three
extensively studied classes of ILP, specifically $n$-fold ILP,
two-stage stochastic ILP and 4\hy block $N$\hy fold
ILP~\cite{DeLoeraHK:AGideasBook,Onn2010}. The distinction lies in the
fact that while in the case of all three previously mentioned special
cases of ILP the ILP matrix must be completely uniform outside of its
global part, here we impose no such restriction. The only part of our
complexity landscape which remains incomplete, the case of mixed
backdoors combined with bounded coefficients, then corresponds to
resolving a challenging open problem in the area of $N$-folds: the
fixed-parameter (in)tractability of 4\hy block $N$\hy fold
ILP~\cite{HemmeckeKW10}. A fixed-parameter algorithm for 4\hy block
$N$\hy fold would also provide significant
algorithmic improvements for problems in areas such as social choice~\cite{KnopKM17}.
We also prove that ILP parameterized by coefficients and a variable or constraint backdoor (i.e., the parameterizations for which we obtain fixed-parameter algorithms) does not admit a polynomial kernel, unless $\NP \subseteq \coNPpoly$.

In the intermediate case of coefficient values encoded in unary (row $2$), we surprisingly show that ILP remains polynomially tractable when the number of global constraints is bounded by a constant, but becomes NP-hard if we use global variables instead. To be precise, we obtain an \XP{} algorithm parameterized by constraint backdoors, rule out the existence of a fixed-parameter algorithm for this case, and also rule out \XP{} algorithms for variable and mixed backdoors. These also represent our most technical results: especially the \XP{} algorithm requires the combination of deep linear-algebraic techniques with tools from the parameterized complexity toolbox.

Last but not least, all our algorithmic results first require us to
compute a fracture backdoor. It turns out that computing fracture
backdoors in ILP is closely related to solving the \textsc{Vertex
  Integrity} problem~\cite{DrangeDregiHof16} on bipartite graphs; unfortunately, while
the problem has been studied on numerous graph classes including
cobipartite graphs, its complexity remained open on bipartite graphs.
Here we obtain both an exact fixed-parameter algorithm as well as a
polynomial time approximation algorithm for finding fracture
backdoors. As an additional result, we also show that the problem is
NP-complete using a novel reduction.

The paper is structured as follows. After introducing the necessary
notions and notation in the preliminaries, we proceed to formally
define our parameter and develop algorithms for computing the desired
backdoors. We then present our results separated by the type of
restrictions put on the size of the matrix coefficients in the remaining sections.

\subsection{Related and Follow-up Work}

This paper represents a natural continuation of previous work aimed at identifying new classes of integer linear programs can be solved efficiently via the use of decompositional parameters which take into account the structure of variable-constraint interactions~\cite{GanianO18,GanianOrdyniakRamanujan17,JansenK15}. However, efforts to characterize tractable classes of ILPs and obtain algorithms with better worst-case runtime guarantees for the problem date back to the classical works of Papadimitriou~\cite{Papadimitriou81}, Lenstra~\cite{Lenstra83} and others~\cite{Kannan87,FrankTardos87}.

At this time, the use of decompositional parameters for ILP remains a dynamic research direction.
First, Chen and Marx~\cite{ChenM18} used Graver-basis driven approach for block structured matrices, the so-called \emph{tree-fold} ILPs, and showed that these ILPs are fpt for the combined parameter that, among others, contains the depth of the tree associated with the constraint matrix.
Indeed, follow-up work by Kouteck\' y, Levin and Onn that appeared in the conference proceedings of \emph{ICALP}~\cite{KouteckyLO18} after the presentation of this paper generalized Corollary~\ref{cor:parvarconFPT} by using Graver-best oracles.
Independently, Eisenbrand, Hunkenschröder and Klein also generalized Corollary~\ref{cor:parvarconFPT} in their recent work, which also appeared at \emph{ICALP}~\cite{EisenbrandHK18}; see also the full joint version on the arXiv repository~\cite{abs-1904-01361}.

A detailed overview of the recent developments in the use of decompositional and structural parameters for solving ILP and its variants can be found in the recent survey dedicated to this topic~\cite{GanianO19}; see also a related survey on block structured matrices of Chen~\cite{Chen2019}.


%
%
%

\newcommand{\ilp}{\textsc{ILP}}
\newcommand{\ilpf}{\textsc{ILP-Fea}}
\newcommand{\milp}{\textsc{MILP}}

\section{Preliminaries}
\label{sec:prelim}

We will use standard graph terminology, see for
instance the textbook by Diestel~\cite{Diestel12}.
In the following let $\mathbf{A}$ be a $n \times m$ matrix and let $C$
and $R$ be a subset of columns and rows of $\mathbf{A}$, respectively. We
denote by $\mathbf{A}_{(R,C)}$ the submatrix of $\mathbf{A}$ restricted
to the columns in $C$ and the rows in $R$. We also denote by
$\mathbf{A}_{(*,C)}$ and $\mathbf{A}_{(R,*)}$ the submatrix of
$\mathbf{A}$ restricted to the columns in $C$ and the
submatrix of $\mathbf{A}$ restricted to the rows in $R$, respectively. We denote by
$\coA$ the maximum absolute value of any entry of $\mathbf{A}$ and by
$\det(\mathbf{A})$ the determinant of $\mathbf{A}$.
For
a vector $\veb$ of size $n$, we will use $\veb[i]$ to denote its
$i$-th entry and we denote by $\coB$ the maximum absolute value of any
entry of $\veb$. We will also use the two following well-known facts~\cite{schrijver1998theory}.

\begin{PRO}\label{pro:det}
 Let $\mathbf{A}$ be an integer $k\times k$ matrix.
Then $\det(\mathbf{A})$ is integer and $|\det(\mathbf{A})| \leq
 k!\Pi_{1\leq i \leq k}\coAi$.
\end{PRO}


\begin{PRO}[Cramer's rule]\label{pro:crammer}
  Let $\mathbf{A}$ be a $k\times k$ non-singular  (i.e., with non-zero
  determinant) matrix and $\veb$ a vector. Then the equation
  $\mathbf{A}\vex=\veb$ has a unique solution such that $\vex[i]=
  \frac{\det(\mathbf{A}(i))}{\det(\mathbf{A})},$ where $\mathbf{A}(i)$
  is the matrix formed by replacing the $i$-th column of $\mathbf{A}$ with
  the vector $\veb$.
\end{PRO}

\subsection{Integer Linear Programming}
\label{sub:ilp}

For our purposes, it will be useful to consider ILP instances which
are in \emph{equation form}. Formally, let an ILP instance $\III$ be a
tuple $(\mathbf{A}, \vex, \veb, \ver, \veu,\eval)$, where:
\begin{itemize}
\item $\mathbf{A}$ is a $n\times m$ matrix of integers (the \emph{constraint matrix}),
\item $\vex$ is a vector of \emph{variables} of size $m$,
\item $\veb$ is an integer vector of size $m$ (the \emph{right-hand side}),
\item $\ver, \veu$ are vectors of elements of $\mathbb{Z}\cup \{\pm \infty \}$ (the \emph{lower} and \emph{upper bounds}, respectively), and
\item $\eval$ is an integer vector of size $m$ (the \emph{optimization function}).
\end{itemize}
Let $A$ be the $i$-th row of $\mathbf{A}$; then we will call $A\vex=\veb[i]$ a \emph{constraint} of $\III$. We will use $\var(\III)$ to denote the set of \emph{variables} (i.e., the elements of $\vex$), and $\cF(\III)$ (or just $\cF$) to denote the set of constraints.
For a subset $U$ of $\var(\III) \cup \cF(\III)$, we denote by $C(U)$
the columns of $\mathbf{A}$ corresponding to variables in $U$ and by
$R(U)$ the rows of $\mathbf{A}$ corresponding to constraints in $U$.

A (partial) assignment $\alpha$ is a mapping from some subset of $\var(\III)$,
denoted by $\var(\alpha)$,  to $\mathbb{Z}$.
An assignment $\alpha$ is called \emph{feasible} if
  \begin{enumerate}
  \item satisfies every constraint in $\cF$, i.e., if $A \alpha(\vex)=\veb[i]$ for each $i$-th row $A$ of $\mathbf{A}$, and
  \item satisfies all the upper and lower bounds, i.e., $\ver[i]\leq \alpha(\vex[i])\leq \veu[i]$.
  \end{enumerate}
Furthermore, $\alpha$ is called a \emph{solution} if the value of $\eval \alpha(\vex)$ is maximized over all feasible assignments; observe that the existence of a feasible assignment does not guarantee the existence of a solution (there may exist an infinite sequence of feasible assignments $\alpha$ with increasing values of $\eval \alpha(\vex)$; in this case, we speak of \emph{unbounded} instances).
Given an instance $\III$, the task in the ILP problem is to compute a
solution for $\III$ or correctly determine that no solution exists.
We remark that other formulations of ILP exist (e.g., a set of inequalities over variables); it is well-known that these are equivalent and can be transformed into each other in polynomial time~\cite{schrijver1998theory}. Moreover, such transformations will only change our parameters (defined in Section~\ref{sec:fracture}) by a constant factor.

Aside from general integer linear programming, we will also be
concerned with two subclasses of the problem.
  \begin{enumerate}
  \item \textsc{ILP-feasibility} is formulated equivalently as ILP, with the restriction that $\eval$ must be the $0$-vector. All hardness results for \textsc{ILP-feasibility} immediately carry over to \textsc{ILP}.
  \item \textsc{Unary ILP} is the class of all ILP instances which are supplied in a unary bit encoding; in other words, the input size of \textsc{Unary ILP} upper-bounds not only the number of variables and constraints, but also the absolute values of all numbers in the input. \textsc{Unary ILP} remains NP-complete in general, but in our setting there will be cases where its complexity will differ from general ILP.
\end{enumerate}
Combining both restrictions gives rise to \textsc{Unary ILP-feasibility}.

There are several ways of naturally representing ILP instances as graphs.
The representation that will be most useful for our purposes will be the so-called \emph{incidence graph}: the incidence graph $G_\III$ of an ILP instance $\III$ is the graph whose vertex set is $\var(\III)\cup \cF(\III)$ and two vertices $s,t$ are adjacent iff $s\in \var(\III)$, $t\in \cF$ and $s$ occurs in $t$ with a non-zero coefficient. An instance $\III'$ is a \emph{connected component} of $\III$ if it is the subinstance of $\III$ corresponding to a connected component of $G_\III$; formally, $\cF(\III')\subseteq \cF(\III)$ is the set of constraints that occur in a connected component of $G_\III$ and $\eval(\III')$ is the restriction of $\eval(\III)$ to $\var(\cF(\III'))$.
For a set $Z\subseteq \cF(\III)\cup \var(\III)$, we will also use $\III\setminus Z$ to denote the ILP instance obtained by removing all constraints in $Z$ from $\cF(\III)$ and removing all variables in $Z$ from all constraints in $\cF(\III)\setminus Z$ and from~$\eval$.

\subsection{Parameterized Complexity}
\label{sub:pc}

In parameterized
algorithmics~\cite{book/FlumG06,book/Niedermeier06,DowneyFellows13}
the runtime of an algorithm is studied with respect to a parameter
$k\in\nn$ and input size~$n$.
The basic idea is to find a parameter that describes the structure of
the instance such that the combinatorial explosion can be confined to
this parameter.
In this respect, the most favorable complexity class is \FPT
(\emph{fixed-parameter tractable})
which contains all problems that can be decided by an algorithm
running in time $f(k)\cdot n^{\bigO{1}}$, where $f$ is a computable
function.
Problems that can be solved in this time are called \emph{fixed-parameter
  tractable} (fpt).

To obtain our lower bounds, we will need the notion of a parameterized reduction.
Formally, a {\em parameterized problem\/} is a subset of $\Sigma^*\times\nn$,
where $\Sigma$ is the input alphabet.
Let $L_1\subseteq \Sigma_1^*\times\nn$ and $L_2\subseteq \Sigma_2^*\times\nn$ be parameterized problems.
A \textit{parameterized reduction} (or fpt-reduction) from $L_1$ to $L_2$ is a mapping $P:\Sigma_1^*\times\nn\rightarrow\Sigma_2^*\times\nn$ such that
\begin{inparaenum}[(i)]
  \item $(x,k)\in L_1$ iff $P(x,k)\in L_2$,
  \item the mapping can be computed by an fpt-algorithm w.r.t.\ parameter $k$, and
  \item there is a computable function $g$ such that $k'\leq g(k)$, where $(x',k')=P(x,k)$.
\end{inparaenum}

A topic related to fixed-parameter algorithms is \emph{kernelization}.
We say that the parameterized problem $L \subseteq \Sigma^* \times \nn$ admits a kernel if there is a function $K: \Sigma^* \times \nn \to \Sigma^* \times \nn$ computable in polynomial time such that $(x,k) \in L$ if and only if $K(x,k) \in L$ and $|K(x,k)| \leq h(k)$ for some computable function $h$.
Informally, a kernel of $L$ is a polynomial time algorithm which given an instance of $L$ produces an equivalent instance of $L$ whose size is bounded by a function of parameter of the original instance.
We can understand a kernel as an effective preprocessing of an instance of some problem.
It is well-known~\cite{CyganFKLMPPS15} that a problem admits a fixed-parameter algorithm if and only if it admits a kernel.
Thus, there is an interest in polynomial kernels, i.e., kernels for which the function $h$ is polynomial.
For proving that our problems do not admit a polynomial kernel we use polynomial parameter transformations from other problems which do not admit a polynomial kernel.

\begin{DEF}[\cite{CyganFKLMPPS15}]
Let $P,Q \subseteq \Sigma^* \times \nn$ be two parameterized problems.
An algorithm $\cA$ is called a \emph{polynomial parameter transformation (PPT)} from $P$ to $Q$ if given an instance $(x,k)$ of problem $P$, $\cA$ works in polynomial time and outputs an equivalent instance $(y,\ell)$ of problem $Q$, i.e., $(x,k) \in P$ if and only if $(y,\ell) \in Q$, such that $\ell \leq p(k)$ for some polynomial $p(\cdot)$.
\end{DEF}

It is known that if a parameterized problem $P$ does not admit a polynomial kernel and there is a PPT from $P$ to $Q$, then $Q$ does not admit a polynomial kernel either~\cite{CyganFKLMPPS15}.

Next, we will define the complexity classes needed to describe our lower bounds.
The class $\W{1}$ captures parameterized intractability and contains
all problems that are fpt-reducible to \probfont{Independent Set} when
parameterized by the size of the solution.
The following relations between the parameterized complexity classes
hold: $\FPT \subseteq \W{1}\subseteq \XP$, where the class \XP
contains all problems solvable in time $\bigO{n^{f(k)}}$ for a
computable function $f$.
Showing $\W{1}$-hardness for a problem rules out the existence of an
fpt-algorithm under standard complexity assumptions.

The class \paraNP
is defined as the class of problems
that are solvable by a non-deterministic Turing machine in fpt time.
In our \paraNP-hardness proofs, we will make use of the following
characterization of \paraNP-hardness given in the book by Flum and Grohe~\cite{book/FlumG06},
Theorem 2.14: any parameterized problem that remains \NP{}\hy{}hard when the
parameter is set to some constant is \paraNP-hard. For problems in \NP, we have
$\W{1}\subseteq \paraNP$ and in particular showing \paraNP-hardness rules out the existence of algorithms with a running time of $\bigO{n^{f(k)}}$.
For our algorithms, we will use the following result as a subroutine. Note that this is a streamlined version of the original statement of the theorem, as used in the area of parameterized algorithms~\cite{FellowsLokshtanovMisraRS08}.

\begin{PRO}[\cite{Lenstra83,Kannan87,FrankTardos87}]
\label{pro:pilp}
There is an algorithm that solves an input ILP instance $\III=(\cF,\eval)$ in time
$p^{\bigoh(p)}\cdot |\III|$, where $p=|\var(\III)|$.
\end{PRO}

\subsection{ILP with Structured Matrices}
\label{sub:ILPwStructure}
Our results build on and extend the classical variable-dimension ILP techniques detailed for instance in the work of Onn and De Loera et al.~\cite{DeLoeraHK:AGideasBook,Onn2010,HemmeckeKW10}. Below, we provide a basic introduction to these techniques and related results.
Let ${\mathbf{A} = \begin{pmatrix}
\mathbf{A}_1 & \mathbf{A}_2 \\
\mathbf{A}_3 & \mathbf{A}_4
\end{pmatrix}}$
be a $2\times 2$ block integer matrix. The \emph{$N$\hy fold 4-block product of $\mathbf{A}$} (denoted by $\mathbf{A}^{(N)}$) is the following integer matrix
\begin{equation*}
\mathbf{A}^{(N)}\quad=\quad
\left(
\begin{array}{ccccc}
  \mathbf{A}_1	& \mathbf{A}_2    & \mathbf{A}_2    & \cdots & \mathbf{A}_2    \\
  \mathbf{A}_3	& \mathbf{A}_4    & 0      & \cdots & 0      \\
  \mathbf{A}_3	& 0      & \mathbf{A}_4    & \cdots & 0      \\
 \vdots & \vdots & \vdots & \ddots & \vdots \\
  \mathbf{A}_3	& 0      & 0      & \cdots & \mathbf{A}_4    \\
\end{array}
\right).
\end{equation*}
Here $\mathbf{A}_1$ is an $r\times s$ matrix, $\mathbf{A}_2$ is an $r\times t$ matrix, $\mathbf{A}_3$ is an $u\times s$ matrix, and $\mathbf{A}_4$ is an $u\times t$ matrix; for convenience, we let $b_{\mathbf{A}}=\max(r,s,t,u)$.
We call an instance $(\mathbf{A}, \vex, \veb, \ver, \veu,\eval)$ of ILP an \emph{$N$\hy fold 4-block} if $\mathbf{A}$ is an $N$\hy fold 4-block product of some $2\times 2$ block integer matrix.
Observe that in such instances the vector $\vex$ is naturally partitioned into a global part (consisting of $s$ variables) and a local part.

\begin{theorem}[\citeauthor{HemmeckeKW10}
  \citeyear{HemmeckeKW10}]\label{thm:nfold4blockip}
  Let $a$ and $z$ be constants and let $\III$ be an \emph{$N$\hy fold
    4-block} ILP instance with $\coA\leq a$, $b_{\mathbf{A}}\leq z$, then $\III$
  can be solved in polynomial time.

\end{theorem}
In the parameterized complexity setting, the above theorem yields an \XP algorithm solving ILP parameterized by $\max(b_{\mathbf{A}},\coA)$ if the matrix is a $N$\hy fold 4-block product.
We note that the existence of a fixed-parameter algorithm for this problem remains a challenging open problem~\cite{HemmeckeKW10}.
However, the problem is known to be fixed-parameter
tractable when either $\mathbf{A}_1$ and $\mathbf{A}_3$ or $\mathbf{A}_1$ and $\mathbf{A}_2$ are omitted;
these variants are called the \emph{$N$-fold ILP} problem and the
\emph{$2$-stage stochastic ILP} problem, respectively.




\begin{theorem}[\citeauthor{HemmeckeOR:13}
  \citeyear{HemmeckeOR:13}, \citeauthor{DeLoeraHK:AGideasBook}
  \citeyear{DeLoeraHK:AGideasBook}]\label{thm:nfoldand2stocip} 
  $N$-fold ILP and $2$-stage stochastic ILP are fpt parameterized by $\coA$
  and $b_{\mathbf{A}}$.
\end{theorem}


\section{The Fracture Number}
\label{sec:fracture}
We are now ready to formally introduce the studied parameter and
related notions. An ILP instance $\III$ is called \emph{$\ell$\hy
  compact} if each connected component of $\III$ contains at most
$\ell$ variables and constraints; equivalently, each connected
component of $G_\III$ contains at most $\ell$ vertices. It is not
difficult to observe that any $\ell$-compact ILP instance can be
solved in time at most $\ell^{\bigoh(\ell)}\cdot |\III|$ due to
  Proposition~\ref{pro:pilp}; indeed,
we can compute a solution for $\III$ by combining solutions for each connected component of $\III$, and hence it suffices to apply Proposition~\ref{pro:pilp} independently on each component.

A set $Z\subseteq \cF\cup \var(\III)$ is called a \emph{backdoor to $\ell$\hy compactness} if $\III\setminus Z$ is $\ell$\hy compact; moreover, if $Z\cap \cF=\emptyset$ then $Z$ is called a \emph{variable-backdoor to $\ell$\hy compactness}, and if $Z\cap \var(\III)=\emptyset$ then $Z$ is a \emph{constraint\hy backdoor to $\ell$\hy compactness}. We use $b_\ell(\III)$ to denote the cardinality of a minimum backdoor to $\ell$\hy compactness, and similarly $b^V_\ell(\III)$ and $b^C_\ell(\III)$ for variable-backdoors and constraint-backdoors to $\ell$\hy compactness, respectively. It is easy to see that, depending on the instance, $b^V_\ell(\III)$ can be arbitrarily larger or smaller than $b^C_\ell(\III)$. On the other hand, $b_\ell(\III)\leq \min(b^V_\ell(\III),b^C_\ell(\III))$.

Clearly, the choice of $\ell$ has a major impact on the size of backdoors to $\ell$-compactness; in particular, $b_\ell(\III)$ could be arbitrarily larger than $b_{\ell+1}(\III)$, and the same of course also holds for variable- and constraint-backdoors. Since we will be interested in dealing with cases where both $\ell$ and $b_\ell(\III)$ are small, we will introduce the \emph{fracture number} $\fr$ which provides bounds on both $\ell$ and $b_\ell$; in particular, we let $\fr(\III)=\min_{\ell\in \Nat} (\max(\ell, b_\ell(\III)))$. Furthermore, we say that a backdoor $Z$ \emph{witnesses} $\fr(\III)$ if $|Z|\leq \fr(\III)$ and $\III\setminus Z$ is $\fr(\III)$-compact. We define $\fr^C(\III)$ and $\fr^V(\III)$ similarly, with $b_\ell(\III)$ replaced by $b^C_\ell(\III)$ and $b^V_\ell(\III)$, respectively. If the instance $\III$ is clear from the context, we omit the reference to $\III$; see Figure~\ref{fig:example} for an example.


We remark that the fracture number represents a strict
generalization of the parameter $b_{\mathbf{A}}$ used in
Theorems~\ref{thm:nfold4blockip} and \ref{thm:nfoldand2stocip}; in particular, $\fr\leq 2b_{\mathbf{A}}$ (and similarly for $\fr^V$ and $\fr^C$ for the latter two theorems). Moreover, the fracture number is well-defined for all ILP instances, not only for $N$-fold $4$-block products. In this respect, $N$-fold $4$-block products with bounded $b_{\mathbf{A}}$ form the subclass of instances with bounded $\fr$ such that each component \emph{must contain precisely the same submatrix}. It is not difficult to see that this is indeed a very strong restriction.

\begin{figure}
 \begin{tabbing}
 \quad \quad \quad \quad \= \kill

\> $\mathtt{maximize \sum_{i=1}^{7}i\cdot x_i}$, \hspace{0.6cm} \text{where}\\

\> $\mathtt{\sum_{i=1}^{7}x_i=32}$, \\

\> $\mathtt{1x_1+y=6}$, \hspace{1.4cm} $\mathtt{2x_2+y=9}$, \\

\> $\mathtt{3x_3+y=14}$, \hspace{1.2cm} $\mathtt{4x_4+y=21}$, \\

\> $\mathtt{5x_5+y=30}$, \hspace{1.2cm} $\mathtt{6x_6+y=41}$.
\end{tabbing}

\caption{The constraints and optimization function of
  a simple ILP instance with $\fr=2$, witnessed by a backdoor containing $y$ and the first constraint.}
\label{fig:example}

\vspace{-0.2cm}
\end{figure}



\section{Computing the Fracture Number}

Our evaluation algorithms for ILP require a backdoor set as part of
their input. In this section we show how to efficiently compute small
backdoor sets, i.e., we show how to solve the following problem.
\problembox{\textsc{Fracture} \COMPFRAC{} (\COMPFRACS{})}
{An ILP instance $\III$ and a natural number $k$.}
{$k$}
{Determine whether $\fr(\III) \leq k$ and if so output a backdoor set witnessing this.}
We also define the variants \VCOMPFRACS{} and \CCOMPFRACS{} that
are concerned with finding a variable or a constraint backdoor,
respectively, in the natural way. Observe that at its core the above problem and its variants are really
a problem on the incidence graph of the ILP instance. Namely, the problems
can be equivalently stated as the following graph problem.

\problembox{\VDTSC{} (\VDTSCS{})}
{An undirected bipartite graph $G$ with bipartition $\{U,W\}$, a set
  $D \in \{U,V(G)\}$, and an integer $k$.}
{$k$}
{Is there a set $B \subseteq D$ of at most $k$ vertices such that
  every connected component of $G \setminus B$ has size at most $k$?}

\noindent It is worth noting that this graph problem is closely related to the
so-called \textsc{Vertex Integrity} problem, which has been studied on
a variety of graph classes, including co-bipartite graphs~\cite{DrangeDregiHof16}.
Unfortunately, to the best of our knowledge nothing is
known about its complexity on bipartite graphs.

To see that each variant of \COMPFRACS{} is equivalent to a specific subcase of the \VDTSCS{} problem (in particular depending on the choice of $D$ in the instance),
consider the following
polynomial-time reductions in both directions. Given an instance
$(\III,k)$ of \COMPFRACS{}, then the instance
$(G_\III,V(G_\III),k)$ of \VDTSCS{} is easily seen to be equivalent.
Similarly, if $(\III,k)$ is an instance of \VCOMPFRACS{} or C\hy
\COMPFRACS{}, then $(G_\III,\var(\III),k)$ and $(G_\III,\cF(\III),k)$
are equivalent instances of \VDTSCS{}.
Moreover, if $I=(G,V(G),k)$ is an instance of \VDTSCS{},
then $(\III,k)$, where $\III$ is any ILP instance such that
$G_\III$ is isomorphic with $G$ is an equivalent instance of
\COMPFRACS{}. Similarly, if $I=(G,U,k)$ is an instance
of \VDTSCS{}, then $(\III,k)$, where $\III$ is any ILP instance such that
$G_\III$ is isomorphic with $G$ and $\var(\III)=U$, is an equivalent instance of
\VCOMPFRACS{}.
Note that such an instance $\III$ can for instance be obtained as
follows:
\begin{itemize}
\item for every vertex $v \in U$, $\III$ has one variable $v$ with
  arbitrary domain,
\item for every vertex $v \in W$, $\III$ has one constraint with
  arbitrary non-zero coefficients on the variables in $N_G(v)$,
\end{itemize}

To justify a parameterized complexity analysis of our detection problems, we first show \NP{}\hy
completeness of our problems. It is worth noting that the \NP{}\hy completeness of \VDTSC{} was far from obvious at first glance due to the restriction to bipartite graphs; indeed, for instance the related problem of deleting at most $k$ vertices such that the remaining graph only contains isolated vertices (\textsc{Vertex Cover}) is well-known to be polynomial on bipartite graphs.

\begin{THE}
  \COMPFRACS{}, \VCOMPFRACS{}, and \CCOMPFRACS{} are \NP{}\hy complete.
\end{THE}
  \begin{proof}
  Because of the equivalence between \COMPFRACS{}, V\hy
  \COMPFRACS{}, \CCOMPFRACS{} and the \VDTSCS{} problem, it is
  sufficient to show that \VDTSCS{} is \NP{}\hy complete for both choices of $D$.
  Because any solution to \VDTSCS{} can be verified in
  polynomial time, it holds that \VDTSCS{} is in \NP{}. Towards
  showing \NP{}\hy hardness of \VDTSCS{} we give a polynomial-time
  reduction from a known variant of the \TSAT{} problem.
  Given a 3-CNF formula $\Phi$
  with variables $x_1,\dotsc,x_n$ and clauses $C_1,\dotsc,C_m$ such
  that every literal occurs in exactly two clauses (this variant of
  \TSAT{} is known to be \NP{}\hy complete~\cite{GareyJohnson79}), we
  construct the instance $\tuple{G,D,k}$ of \VDTSCS{} as
  follows. We set $k=n+2m$ and the graph $G$ will be the disjoint
  union of certain variable and clause gadgets introduced below plus connections
  between these variable and clauses gadgets. Namely, for every variable
  $x_i$, the graph $G$ contains the variable gadget $G(x_i)$ with the following
  vertices and edges:
  \begin{itemize}
  \item two vertices $x_i$ and $\overline{x_i}$,
  \item $k-5$ vertices $c_i^1,\dotsc,c_i^{k-5}$,
  \item for every $j$ with $1 \leq j \leq k-5$ the two edges
    $\{x_i,c_i^j\}$ and $\{\overline{x_i},c_i^j\}$.
  \end{itemize}
  Moreover for every clause $C_j$ of $\Phi$ with literals $l^1_j$,
  $l^2_j$, $l^3_j$, the graph $G$ contains
  a clause gadget $G(C_J)$ with the following vertices and edges:
  \begin{itemize}
  \item three vertices $l_j^1$, $l_j^2$, and $l_j^3$,
  \item $k-3$ vertices $b_j^1,\dotsc,b_j^{k-3}$,
  \item for every $i$ with $1\leq i \leq k-3$ the three edges
    $\{b_j^i,l_j^1\}$, $\{b_j^i,l_j^2\}$, and $\{b_j^i,l_j^3\}$.
  \end{itemize}
  Note that $G(C_j)$ is simple a complete bipartite graph
  with bipartition \mbox{$\{\{l^1_j,l^2_j,l^3_j\}, \SB b_j^i \SM 1 \leq i\leq k-3\SE$}.
  Now~$G$ consists of the disjoint union of $G(x_1), \dotsc, G(x_n),
  G(C_1),\dotsc, G(C_m)$ plus the following vertices and edges, which
  ensure the required connections between the variable and clause gadgets:
  \begin{itemize}
  \item For every clause $C_j$ (for some $j$ with $1 \leq j \leq m$)
    with literals $l_j^1$, $l^2_j$, and $l^3_j$ and every $a \in
    \{1,2,3\}$ we add the vertices $d_j^a$ and $e_j^a$ and the edges
    $\{l_j^a,d_j^a\}$ and $\{l_j^a,e_j^a\}$ to $G$. Moreover, if
    $l_j^a=x_i$ for some $i$ with $1 \leq i \leq n$, we additionally
    add the edges $\{x_i,d_j^a\}$ and $\{x_i,e_j^a\}$ to $G$ and if on
    the other hand $l_j^a=\overline{x_i}$ for some $i$ as above, then
    we add the edges $\{\overline{x_i},d_j^a\}$ and
    $\{\overline{x_i},e_j^a\}$ to $G$.
  \end{itemize}
 	 \begin{figure}
	\centering
	\usetikzlibrary{fit,calc,shapes}

\begin{tikzpicture}[node distance=.6cm]
  \tikzstyle{vertex}=[draw, circle]
  \tikzstyle{vertexNoimp}=[draw, circle, fill, inner sep=2pt]
  \tikzstyle{fitting}=[draw, rounded corners, inner sep=1.35em]
  \tikzstyle{edge}=[thick]

  \begin{scope}[yshift=-.2cm, node distance=1.5cm]
    \node[vertex, label={90:$\ell_j^1$}] (lj1) {};
    \node[vertex, below of=lj1, label={90:$\ell_j^2$}] (lj2) {};
    \node[vertex, below of=lj2, label={90:$\ell_j^3$}] (lj3) {};

    \node[vertexNoimp, left of=lj1, label={180:$q_j^1$},yshift=.5cm] (qj1) {};
    \node[vertexNoimp, below of=qj1, label={180:$q_j^2$}, yshift=.9cm] (qj2) {};
    \node[vertexNoimp, below of=qj2, label={180:$q_j^{k-2}$}, yshift=-1.5cm] (qj3) {};
    \node at ($(qj2)!.44!(qj3)$) {$\vdots$};

    \foreach \i in {1,2,3} {
        \foreach \j in {1,2,3} {
              \draw[edge] (lj\j) to (qj\i);
        }
    }

    \node[fitting, fit=(lj1)(lj2)(lj3)(qj1), xshift=-12pt, label={90:$C_j$}] {};
  \end{scope}

  \begin{scope}[xshift=4cm, yshift=1.7cm]
    \node[vertex, label={90:$x_1$}] (x1) {};
    \node[vertex, below of=x1, label={270:$\overline{x_1}$}, yshift=-.6cm] (non_x1) {};
    \node[vertexNoimp, right of=x1, label={0:$p_1^1$}, yshift=.2cm, xshift=.2cm] (p11) {};
    \node[vertexNoimp, below of=p11, label={0:$p_1^2$}] (p12) {};
    \node[vertexNoimp, below of=p12, label={0:$p_1^{\ell - 3}$}, yshift=-.4cm] (p13) {};
    \node at ($(p12)!.4!(p13)$) {$\vdots$};
    \foreach \i in {1,2,3} {
        \draw[edge] (x1) to (p1\i);
        \draw[edge] (non_x1) to (p1\i);
   }
   \node[fitting, fit=(x1)(non_x1)(p11)(p12)(p13), xshift=12pt] {};
  \end{scope}

  \begin{scope}[xshift=4cm, yshift=-1.5cm]
    \node[vertex, label={90:$x_5$}] (x5) {};
    \node[vertex, below of=x5, label={270:$\overline{x_5}$}, yshift=-.6cm] (non_x5) {};
    \node[vertexNoimp, right of=x5, label={0:$p_5^1$}, yshift=.2cm, xshift=.2cm] (p51) {};
    \node[vertexNoimp, below of=p51, label={0:$p_5^2$}] (p52) {};
    \node[vertexNoimp, below of=p52, label={0:$p_5^{\ell - 3}$}, yshift=-.4cm] (p53) {};
    \node at ($(p52)!.4!(p53)$) {$\vdots$};
    \foreach \i in {1,2,3} {
        \draw[edge] (x5) to (p5\i);
        \draw[edge] (non_x5) to (p5\i);
   }
   \node[fitting, fit=(x5)(non_x5)(p51)(p52)(p53), xshift=12pt] {};
  \end{scope}

  \begin{scope}[xshift=4cm, yshift=-4.7cm]
    \node[vertex, label={90:$x_7$}] (x7) {};
    \node[vertex, below of=x7, label={270:$\overline{x_7}$}, yshift=-.6cm] (non_x7) {};
    \node[vertexNoimp, right of=x7, label={0:$p_7^1$}, yshift=.2cm, xshift=.2cm] (p71) {};
    \node[vertexNoimp, below of=p71, label={0:$p_7^2$}] (p72) {};
    \node[vertexNoimp, below of=p72, label={0:$p_7^{\ell - 3}$}, yshift=-.4cm] (p73) {};
    \node at ($(p72)!.4!(p73)$) {$\vdots$};
    \foreach \i in {1,2,3} {
        \draw[edge] (x7) to (p7\i);
        \draw[edge] (non_x7) to (p7\i);
   }
   \node[fitting, fit=(x7)(non_x7)(p71)(p72)(p73), xshift=12pt] {};
  \end{scope}

  \begin{scope}[xshift=2cm, yshift=1.5cm]
    \node[vertex, label={90:$d_j^1$}] (dj1) {};
    \node[vertex, label={270:$e_j^1$}, below of=dj1] (ej1) {};
  \end{scope}

  \begin{scope}[xshift=2cm, yshift=-1.5cm]
    \node[vertex, label={90:$d_j^2$}] (dj2) {};
    \node[vertex, label={270:$e_j^2$}, below of=dj2] (ej2) {};
  \end{scope}

  \begin{scope}[xshift=2cm, yshift=-4.5cm]
    \node[vertex, label={90:$d_j^3$}] (dj3) {};
    \node[vertex, label={270:$e_j^3$}, below of=dj3] (ej3) {};
  \end{scope}

  \foreach \a in {1,2,3} {
    \draw[edge] (dj\a) -- (lj\a) -- (ej\a);
  }

  \draw[edge] (dj1) -- (x1) -- (ej1);
  \draw[edge] (dj2) -- (non_x5) -- (ej2);
  \draw[edge] (dj3) -- (x7) -- (ej3);
\end{tikzpicture}
	\caption{The interaction between clause and vertex gadgets for the clause \mbox{$C_j=x_1\vee \overline{x_5}\vee x_7$}.}\label{fig:clause_var_interaction}
\end{figure}

  See Figure~\ref{fig:clause_var_interaction} for example of variable and clause gadgets and how they are connected.
  This completes the construction of $G$, which is clearly bipartite
  as for instance witnessed by the bipartition \mbox{$\{U,V(G) \setminus U\}$}, where
  \mbox{$U=\SB x_i,\overline{x_i},l_j^1,l_j^2,l_j^3 \SM 1 \leq i
  \leq n \land 1\leq j \leq m\SE$}. We will show below that there is always
  a solution that is entirely contained in $U$, which implies
  that the hardness result holds for $D \in \{U, V(G)\}$, and hence
  all versions of the fracture backdoor set problem, i.e.,
  $\COMPFRACS{}$, $\VCOMPFRACS{}$, and $\CCOMPFRACS{}$, are \NP{}\hy complete.
  Note that the
  reduction can be computed in polynomial time and it remains to show
  the equivalence between the two instances.

  Towards showing the forward direction, assume that \mbox{$\alpha :
  \{x_1,\dotsc,x_n\} \rightarrow \{0,1\}$} is a satisfying assignment
  for~$\Phi$. Because $\alpha$ satisfies $\Phi$ it follows that for
  every clause $C_j$ with literals $l_j^1$, $l_j^2$, and $l_j^3$
  there is at least one index \mbox{$a(C_j) \in \{1,2,3\}$}
  such that the literal $l_j^{a(C_j)}$ is satisfied by $\alpha$. We
  claim that the set $B$ defined by:
  \begin{itemize}
  \item for every $i$ with $1 \leq i \leq n$, $B$ contains $x_i$ if
    $\alpha(x_i)=1$ and $\overline{x_i}$, otherwise,
  \item for every $j$ with $1 \leq j \leq m$, $B$ contains the
    vertices in $\SB l_j^b \SM b \in \{1,2,3\}\setminus
    \{a(C_j)\}\SE$.
  \end{itemize}
  is a solution for $(G,U,k)$. Because $B$ contains exactly one vertex
  for every variable of $\Phi$ and exactly two vertices for every
  clause of $\Phi$, it holds that $|B|=k=n+2m$, as required. Moreover,
  $B \subseteq U$. It hence only remains to show that every component
  of $G \setminus B$ has size at most $k$. Towards showing this first
  consider a component $C$ of ${G \setminus B}$ that contains at least
  one vertex from a variable gadget $G(x_i)$ for some $i$ with $1\leq
  i \leq n$. Then $G(x_i)\cap B \in \{\{x_i,\overline{x_i}\}\}$ and
  hence $G(x_i) \setminus B$ is connected, which implies that \mbox{$G(x_i)
  \setminus B \subseteq C$}. W.l.o.g. assume that $G(x_i) \cap
  B=\{x_i\}$. Then $\alpha(x_i)=1$ and it follows that all literal
  vertices of clause gadgets that correspond to the literal
  $\overline{x_i}$ are contained in $B$. Since moreover
  $\overline{x_i}$ is contained in exactly two clauses, we obtain that
  $C$ consists of exactly $k-4$ vertices in $G(x_i) \setminus B$ plus
  the four vertices $d_{j_1}^{a_1}$, $e_{j_1}^{a_1}$, $d_{j_2}^{a_2}$
  and $e_{j_2}^{a_2}$ defined by $l_{j_1}^{a_1}=\overline{x_i}$ and
  $l_{j_2}^{a_2}=\overline{x_i}$. Hence in total $C$ contains exactly
  $k$ vertices as required. Now consider a component $C$ that contains
  at least one vertex from a clause gadget $G(C_j)$ for some $j$ with
  $1\leq j \leq m$. Then $|G(C_j) \cap B|=2$ and moreover $B$ contains
  all but exactly one literal vertex say $l_j^a$ for some $a \in
  \{1,2,3\}$ from $G(C_j)$. W.l.o.g. let $x_i$ be the literal of $C_j$
  corresponding to $l_j^a$. Then $\alpha(x_i)=1$ and hence $x_i \in
  B$. It follows that $C$ consists of the exactly $k-2$ vertices in
  $G(C_j)\setminus B$ plus the two vertices $d_j^a$ and $e_j^a$. Hence
  in total $C$ contains exactly $k$ vertices, as required. Because
  every component of $G \setminus B$ that neither contains a vertex
  from a vertex gadget nor from a clause gadget has size exactly one,
  this shows that $B$ is indeed a solution for $(G,U,k)$ and hence
  also for $(G,V(G),k)$.

  Towards showing the reverse direction, let $B$ be a solution for
  $(G,V(G),k)$.
  We first show that w.l.o.g. we can assume that $B
  \subseteq U$. So assume that $B \nsubseteq U$. We distinguish three
  cases: $B$ contains a vertex $d_j^a$ or $e_j^a$ for some $j$ and $a$
  with $1 \leq j \leq m$ and $1 \leq a \leq 3$. Let $u$ and $v$ be the
  two vertices adjacent to $d_j^a$ and $e_j^a$. If $B$ contains both
  $d_j^a$ and $e_j^a$, then it is straightforward to verify that $B
  \setminus \{d_j^a,e_j^a\} \cup \{u,v\}$ is also a solution. So
  assume that $B$ contains only $d_j^a$ (the case that $B$ contains
  only $e_j^a$ is analogous). If $\{u,v\} \subseteq B$, then $B
  \setminus \{d_j^a\}$ is still a solution. Hence assume that
  w.l.o.g. $u \notin B$. But then $(B \setminus \{d_j^a\})\cup \{u\}$
  is a solution. Hence in all cases we could transform $B$ into a
  solution that does not contain a vertex $d_j^a$ or $e_j^a$. Next
  consider the case that $B$ contains some vertex $c_i^j$ for some $i$
  and $j$ with $1 \leq i \leq n$ and $1\leq j \leq k-5$. In this case
  one can use an argumentation very similar to the previous case to
  transform $B$ into a solution not containing such a vertex.
  Hence there only remains the case that $B$ contains some vertex
  $b_j^i$ for some $i$ and $j$ with $1 \leq i \leq k-3$ and $1\leq j
  \leq m$. In this case it is straightforward to verify that removing
  all vertices from $B \cap \{b_j^1,\dotsc,b_j^{k-3}\}$ and replacing
  those with an equal (or less) amount of vertices in $\{l_j^1,l_j^2,l_j^3\}$
  will again give a solution. Hence we can assume that $B \subseteq
  U$.

  We show next that $B$ contains at least one of $x_i$ and
  $\overline{x_i}$ from every variable gadget $G(x_i)$. Suppose not
  and consider the component $C$ of $G \setminus B$ containing $x_i$.
  Because $B \subseteq U$, we obtain that $C$ contains all $k-3$
  vertices in $G(x_i)$ and additionally at least the $8$ vertices adjacent
  to $x_i$ and $\overline{x_i}$. Hence $|C|\geq k-3+8>k$ a
  contradiction to our assumption that $B$ is a solution.

  We show next that $B$ contains at least two of
  $\{l_j^1,l_j^2,l_j^3\}$ from every clause gadget $G(C_j)$. Suppose
  not and consider a component $C$ of $G \setminus B$ containing at
  least one vertex from $G(C_j)$. Because $B \subseteq U$, we obtain
  that $C$ contains all of the at least $k-3+2=k-1$ vertices in
  $G(C_j) \setminus B$ and additionally the at least four vertices
  adjacent to the (at least two) literal vertices in
  $\{l_j^1,l_j^2,l_j^3\}\setminus B$. Hence \mbox{$|C|\geq k-1+4>k$} a
  contradiction to our assumption that $B$ is a solution.

  Hence $B$ contains at least one vertex for every variable of $\Phi$
  and at least two vertices for every clause of $\Phi$. Moreover,
  because $B$ is a solution it holds that $|B|\leq k=n+2m$. Hence
  $|B|=n+2m$ and $B$ contains exactly one vertex from every variable
  gadget and exactly two vertices from every clause gadget. We claim
  that the assignment $\alpha$ with $\alpha(x_i)=1$ if and only if
  $x_i \in B$ is a satisfying assignment for $\Phi$. Suppose not
  and let $C_j$ be a clause of $\Phi$ that is not satisfied by $\Phi$
  and let $l_j^a$ be the (unique) literal vertex of $G(C_j)$ that is
  not in $B$. Consider the component $C$ of $G \setminus B$ that
  contains $l_j^a$ and assume w.l.o.g. that $l_j^a=x_i$ for some $i$
  with $1 \leq i \leq n$. Because $\alpha$ does not satisfy $C_j$, we
  obtain that $x_i \notin B$. Because furthermore $B \subseteq U$ we
  obtain that $C$ contains all of the $k-3+1=k-2$ vertices in $G(C_j)
  \setminus B$ and additionally at least the two vertices
  adjacent to $l_j^a$ as well as the vertex $x_i$. Hence in total $C$
  contains at least $k-2+3>k$ vertices, a contradiction to our
  assumption that $B$ is a solution.
\end{proof}

Even though \COMPFRACS{} is \NP{}\hy complete, here we provide two efficient
algorithms for solving it: we show that the problem is
fixed-parameter tractable parameterized by $k$ and can be
approximated in polynomial time within a factor of $k$. Both of these
algorithms are based on the observation that any backdoor has to
contain at least one vertex from every connected subgraph of the
instance of size $k+1$.
\begin{THE}
\label{thm:det-fpt}
  \COMPFRACS{}, \VCOMPFRACS{}, and \CCOMPFRACS{} can be solved in time
  $\bigO{(k+1)^k|E(G)|}$ and are hence fpt.
\end{THE}

  \begin{proof}
  Because of the equivalence of the problems \COMPFRACS{}, V\hy
  \COMPFRACS{}, and \CCOMPFRACS{} with the \VDTSCS{} problem, it is
  sufficient to show the result for \VDTSCS{}.

  We will show the theorem by providing a depth-bounded search tree
  algorithm for any instance $I=\tuple{G,D,k}$ of \VDTSCS{}, which is based on the following
  observations.
  \begin{itemize}
  \item[O1] If $G$ is not connected then a solution for $I$ can be
    obtained as the disjoint union of solutions for every component of $G$.
  \item[O2] If $G$ is connected and $C$ is any set of $k+1$ vertices
    of $G$ such that $G[C]$ is connected, then any solution for $I$
    has to contain at least one vertex from $C$.
  \end{itemize}
  These observations lead directly to the following recursive algorithm
  that given an instance $I=\tuple{G,D,k}$ of \VDTSCS{} either
  determines that the instance is a \no\hy instance or outputs a
  solution $B \subseteq D$ of minimal size for $I$. The algorithm also
  remembers the maximum size of any component in a global constant
  $c$, which is set to $k$ for the whole duration of the algorithm.
  The algorithm first checks
  whether $G$ is connected. If $G$ is not connected the algorithm
  calls itself recursively on the instance $(C,D\cap C, k)$ for each component
  of $G$. If one of the recursive calls returns \no{} or if the size of
  the union of the solutions returned for each component exceeds $k$,
  the algorithm returns that $I$ is a \no\hy instance. Otherwise the
  algorithm returns the union of the solutions returned for each
  component of $G$.

  If $G$ is connected and $|V(G)|\leq c$, the algorithm returns the
  empty set as a solution. Otherwise, i.e. if $G$ is connected but
  $|V(G)|>c$ the algorithm first computes a set $C$ of $c+1$ vertices
  of $G$ such that $G[C]$ is connected. This can for instance be achieved
  by a depth-first search that starts at any vertex of $G$ and stops
  as soon as $c+1$ vertices have been visited. If $C \cap D=\emptyset$
  then the algorithm returns \no{}. Otherwise the algorithm
  branches on the vertices in $C \cap D$, i.e., for every $v \in C
  \cap D$ the
  algorithm recursively computes a solution for the instance $(G
  \setminus \{v\},k-1)$. It then returns the solution of minimum size
  returned by any of those recursive calls, or \no\hy if none of those
  calls return a solution. This completes the description of the
  algorithm. The correctness of the algorithm follows immediately from
  the above observations. Moreover the running time of the algorithm
  is easily seen to be dominated by the maximum time required for the
  case that at each step of the algorithm $G$ is connected.
  In this case the running time can be obtained as the
  product of the number of branching steps times the time spent on
  each of those. Because at each recursive call the parameter $k$ is
  decreased by at least one and the number of branching choices is at
  most $c+1$, we obtain that there are at most $(c+1)^k=(k+1)^k$ branching
  steps. Furthermore, the time at each branching step is dominated by
  the time required to check whether $G$ is connected, which is linear
  in the number of edges of $G$. Putting everything together,
  we obtain $\bigO{(k+1)^k|E(G)|}$ as the total time required by the
  algorithm, which completes the proof of the lemma.
\end{proof}

We note that the depth-first search algorithm in the above proof can be
easily transformed into a polynomial-time approximation algorithm for
\COMPFRACS{} and its variants that exhibits an approximation ratio of
$k+1$. In particular, instead of branching on the vertices of
a connected subgraph $C$ of $G$ with $k+1$ vertices, this algorithm
would simply add all the vertices of $C$ into the current
solution. This way we obtain:

\begin{THE}
  \COMPFRACS{}, \VCOMPFRACS{}, and \CCOMPFRACS{} can be approximated
  in polynomial time within a factor of $k+1$.
\end{THE}

\section{The Case of Bounded Coefficients}
The goal of this section is to obtain the algorithmic results presented on the first row of Table~\ref{fig:results}. Recall that in this case we will be parameterizing also by $c_{\mathbf{A}}$, which is the maximum absolute coefficient occurring in $\mathbf{A}$.
Before we proceed to the results themselves, we first need to
introduce a natural notion of ``equivalence'' among the components of an ILP instance.

Let $Z$ be a backdoor to $\ell$\hy compactness for an ILP instance~$\III$.
We define the equivalence relation $\sim$ on the components of $\III\setminus Z$ as follows:
two components $C_1$ and $C_2$ are equivalent iff there exists a
bijection $\gamma$ between $\var(C_1)$ and $\var(C_2)$ such that
the ILP instance obtained from $\III$ after renaming the variables in
$\var(C_1)$ and $\var(C_2)$ according to $\gamma$ and $\gamma^{-1}$,
respectively, is equal to $\III$.
We say that components $C_1$  and $C_2$ have the same \emph{type} if
$C_1 \sim C_2$.
\begin{LEM}
\label{lem:4BlockNumberOfEqiuvalenceClasses}
  Let\/ $\III$ be an ILP instance and ${k=\fr(\III)}$.
  For any backdoor witnessing $\fr(\III)$,
  $\sim$ has at most
  ${\bigl(2\coA{}(\III)+1\bigr)}^{2k^2}$ equivalence classes.
  Moreover, one can test whether two components have the same type in
  time $\bigoh(k!k^2)$.
\end{LEM}
\begin{proof}
  Let $Z$ be the backdoor witnessing $\fr(\III)$ and fix a component $C$ of $\III\setminus Z$.
  First observe that there are only 3 submatrices of the constraint matrix of $\III$ that can contain nonzero coefficients and containing an element of $C$; we refer to Fig.~\ref{fig:ILPtransOverview}, where we denote these matrices $\mathbf{Q}_{C}, \mathbf{Q}_C^V$, and $\mathbf{Q}_C^C$.
  We will now bound the number of these possibly nonzero coefficients.
  In order to do this we denote by ${g_v = |\var(Z)|}, {g_c = |Z| - g_v}, {c_v = |\var(C)|}, {c_c = |C| - c_v}$, and ${c = \max{c_c, c_v}}$.
  Observe that $c\le k$.
  Now the number of possibly nonzero coefficients is bounded by $g_vc_c + g_cc_v + c_cc_v \le (g_v + g_c)c + c^2 \le 2k^2$.
  We finish the proof of the first part by observing that the number of possible coefficients is bounded by $2\coA{}(\III) + 1$.

  Observe that two components $C_1$ and $C_2$ have the same type if their number of constraints and variables is the same and there exist a permutation of variables of $C_1$ and a permutation of constraints of $C_1$ such that the three submatrices of $\III$ containing nonzero elements are exactly the same.
  Again as $|C| \le k$ one can check all pairs of permutations in time $k!$ and for each pair we are checking $\bigoh(k^2)$ entries.
\end{proof}

We now proceed to the main tool used for our algorithms.

\begin{THE}
\label{thm:nonuniform4blockFPT}
Let $\overline{\III}$ be an ILP instance with matrix $\overline{\mathbf{A}}$, $Z$ be a backdoor set witnessing $\fr(\overline{\III})$, and let $n$ be the number of components of\/ $\overline{\III}\setminus Z$.
There is an algorithm which runs in time $\bigoh(n^2(\fr(\III)+1)! + |\III|)$ and computes a $(r+u)\times (s+t)$ matrix ${\mathbf{A} = \begin{pmatrix}
\mathbf{A}_1 & \mathbf{A}_2 \\
\mathbf{A}_3 & \mathbf{A}_4
\end{pmatrix}}$, a positive integer $N\le n$, and a $4$-block $N$-fold instance $\III=(\mathbf{A}^{(N)},\vex,\veb, \ver, \veu, \eval)$ such that:
\begin{enumerate}
\item[(P1)] any solution for $\III$ can be transformed (in polynomial time) into a solution for
  $\overline{\III}$ (and vice versa), and
\item[(P2)] $\max\{r,s\}\le\fr(\overline{\III})$ and $\max\{t,u\}\le f\bigl(c_{\mathbf{\overline{A}}},\fr(\overline{\III})\bigr)$ for some computable function~$f$.
\end{enumerate}
\end{THE}
\begin{proof}
Let $\mathcal{C}$ be the set of connected components of $\overline{\III}\setminus Z$.
We define a triple of matrices $(\mathbf{Q}^V_C, \mathbf{Q}^C_C, \mathbf{Q}_C)$ for a component $C$.
Please refer to Figure~\ref{fig:ILPtransOverview}.
\begin{itemize}
  \item The matrix $\mathbf{Q}_C^V$ is the part of constraints in $C$
    dealing with variables in $Z$, that is, $\overline{\mathbf{A}}_{\cF(C),\var(Z)}$,
  \item the matrix $\mathbf{Q}_C^C$ is the part of the constraints in $Z$ dealing with~$\var(C)$, that is, $\overline{\mathbf{A}}_{\cF(Z),{\var(C)}}$, and
  \item the matrix $\mathbf{Q}_C$ is the part of constraints in $C$ dealing with~$\var(C)$, that is, $\overline{\mathbf{A}}_{\cF(C),\var(C)}$.
\end{itemize}

\begin{figure}[t]
  \begin{tikzpicture}[scale=.8]
  \newcommand\YD{3}
  \newcommand\YU{8}
  \newcommand\XD{0}
  \newcommand\XU{6}
  \newcommand\dX{1}
  \newcommand\dY{1}

  \coordinate (S) at (\XD, \YU);
  \coordinate (Rmost) at (\XU, \YU);
  \coordinate (Dmost) at (\XD, \YD);
  \coordinate (podS) at (\XD, \YU-\dY);
  \coordinate (podR) at (\XU, \YU-\dY);
  \coordinate (zaS) at (\XD + \dX, \YU);
  \coordinate (zaD) at (\XD + \dX, \YD);

  \coordinate (Cld) at (\XU,5);
  \coordinate (Clu) at (\XU,6);
  \coordinate (Crd) at (\XD - \dX,5);
  \coordinate (Cru) at (\XD - \dX,6);

  \coordinate (Cul) at (3,\YU + \dY);
  \coordinate (Cur) at (4,\YU + \dY);
  \coordinate (Cdl) at (3,\YD);
  \coordinate (Cdr) at (4,\YD);

  \draw (S) -- (Dmost);
  \draw (S) -- (Rmost);
  \draw (podS) -- (podR);
  \draw (zaS) -- (zaD);

  \draw[dashed] (Cul) -- (Cdl);
  \draw[dashed] (Cur) -- (Cdr);
  \draw[dashed] (Cld) -- (Crd);
  \draw[dashed] (Clu) -- (Cru);

  \node at (3.5,5.5) {$Q_C$};
  \node at (.5,5.5) {$Q^V_C$};
  \node at (3.5,7.5) {$Q^C_C$};

  \node at (3.5,6.5) {$0$};
  \node at (3.5,4) {$0$};
  \node at (5,5.5) {$0$};
  \node at (2,5.5) {$0$};

  \node at (3.5,8.5) {$x_C$};
  \node at (-1,5.5) {con($C$)};
\end{tikzpicture}
  \caption{A situation for a component $C$.}
  \label{fig:ILPtransOverview}
\end{figure}

Observe that this totally decomposes all constraints and variables contained in $C$ as all coefficient for other variables are 0 and variables of $C$ cannot appear in other components.
For a triple of matrices $T = (\mathbf{Q}^V, \mathbf{Q}^C, \mathbf{Q})$ a~component~$C$ has type $T$ if ${\mathbf{Q}^V = \mathbf{Q}_C^V}, {\mathbf{Q}_C = \mathbf{Q}_C^C}$, and ${\mathbf{Q} = \mathbf{Q}_C}$ holds.
The~set of all possible types is the set $$\mathcal{T} = \SB T = (\mathbf{Q}^V, \mathbf{Q}^C, \mathbf{Q}) \SM \exists C\in\mathcal{C} \textrm{ with type } T \SE.$$
The~multiplicity $\mathrm{mult}(T)$ of type $T\in\mathcal{T}$ is the~number of components in $\mathcal{C}$ having type $T$.
We set ${N = \max_{T\in\mathcal{T}}\mathrm{mult}(T)}$.

The~idea of the~proof is to build the matrix $\mathbf{A}_1$ from $Z$ and matrices $\mathbf{A}_2, \mathbf{A}_3, \mathbf{A}_4$ as representatives of the types in such a~way that the resulting $N$-fold 4 block ILP is equivalent to the given ILP instance~$\overline{\III}$.

The matrix $\mathbf{A}_1$ is simply the submatrix of $Z$ that is the part of global constraints of $\overline{\mathbf{A}}$ containing $\var(Z)$ only.
\begin{CLM}\label{clm:typesOfSameMultiplicity}
There is an ILP instance $\hat{\III}$ that is equivalent to ILP instance $\bar{\III}$ with $\mathrm{mult}_{\hat{\III}}(T) = N$ for all $T\in\mathcal{T}_{\hat{\III}}$. Moreover, $c_{\hat{\III}} = c_{\bar{\III}}$ and the sizes of the matrices $\mathbf{Q}$ can only double.
\end{CLM}
In this case we put all possible matrices on a diagonal of the relevant matrix $\mathbf{A}_4$, next to each other in the matrix $\mathbf{A}_2$, and under each other in the matrix $\mathbf{A}_3$.
That is we set $\mathbf{A}_2$ to horizontal concatenation of all $(\mathbf{Q}^C_T)_{T\in \mathcal{T}}$, $\mathbf{A}_3$ to vertical concatenation of $(\mathbf{Q}^V_T)_{T\in\mathcal{T}}$, and finally $\mathbf{A}_4$ has matrices $(\mathbf{Q}_T)_{T\in\mathcal{T}}$ on its diagonal.
The bound on size of the matrix $\mathbf{A}$ follows from Lemma~\ref{lem:4BlockNumberOfEqiuvalenceClasses} and Claim~\ref{clm:typesOfSameMultiplicity}.
\end{proof}

\begin{proof}[Proof of Claim~\ref{clm:typesOfSameMultiplicity}]
The idea here is to take a type with less representatives and add a new one as a copy of a previous one.
But this has to be done carefully in order to maintain equivalence of intermediate ILPs.
For the local part we start by observing that if we add a copy of some previous component, then the~set of solutions for these two components is the~same.
However, as these components also interact with the global constraints we would like to have to restrict~the set of solutions of the~newly added component to all~0~solution only.
Note that this cannot be done using lower and upper bounds only as the former set of solutions does not have to contain such a solution.
That is, the (optimal) setting of global variables together with setting all component local variables to $0$ can violate the right-hand side.
In order to achieve the claim, we extend the matrices we have obtained from the component~$C$ in the following way.
Let $C$ be of type $T = (\mathbf{Q}_C^V, \mathbf{Q}_C^C, \mathbf{Q}_C)$ then the {\em extension of type $T$} if $$\hat{T} = \bigl(\mathbf{Q}_C^V, [\mathbf{Q}_C^C\mid \mathbf{0}], [\mathbf{Q}_C\mid \mathbf{Q}_C]\bigr).$$ 
We denote the former $C$\hy variables as $x_C$ and the new $C$\hy variables as $\hat{x}_C$.
We say that the extension is of
\begin{itemize}
  \item {\em first kind} if $\ell_C\le x_C\le u_C$ and $0\le \hat{x}_C\le 0$, and
  \item {\em second kind} if $0\le x_C\le 0$ and $\ell_C\le \hat{x}_C\le u_C$.
\end{itemize}
Note that with this we have only doubled the number of local variable of component~$C$.

\begin{CLM}\label{clm:extendedFirstKind}
Let $\III$ be an ILP instance and let $T$ be a type of~$\III$. 
Denote $\III_{T\to \hat{T}}$ the ILP instance $\III$ where components of type $T$ are replaced with components of $\hat{T}$ of the first kind.
Then, there is a bijection between solutions of ILP instances $\III$ and $\III_{T\to \hat{T}}$.
\end{CLM}
\begin{proof}
Note that it holds that $\hat{x}_C = 0$ for every component $C$ of type $\hat{T}$.
Now a solution for $\III_{T\to \hat{T}}$ has a natural projection to a solution of $\III$ (forget all $\hat{x}_C$ variables).
Furthermore, a solution for $\III$ can be extended to a solution of $\III_{T\to \hat{T}}$ by setting $\hat{x}_C = 0$ for each component $C$ of type $T$.
This yields a bijection between the solution sets.
\end{proof}
We say that a component $C$ is {\em extended} if it has been created by the extension of the first kind.
We transform all components with multiplicity less than $N$ to extended components and denote $\III_{E}$ the resulting ILP instance.
Note that by Claim~\ref{clm:extendedFirstKind} the ILP instances $\III$ and $\III_{E}$ are in equivalent.

\begin{CLM}
Let $\III$ be an ILP instance, let $C$ be a~component of~$\III$, and let $C'$ be an extension of $C$ of the second kind.
Denote~$\III'$ the ILP instance $\III$ with $C'$ added (i.e., it has one more component) then instances $\III$ and $\III'$ are equivalent.
\end{CLM}
\begin{proof}
First we argue that $\III$ does have a solution if and only if $\III'$ does.
To see this take a solution $\vex$ of $\III$ and let $x_C$ be the part of $\vex$ corresponding to $C$\hy variables.
We build a solution to $\III'$ follows.
We copy the solution of every variable but the variables of $C'$.
We set variables $x_{C'}=0$ and $\hat{x}_{C'} = x_C$.

Note that by this we have actually build a natural correspondence between the set of solutions to $\III$ and the set of solutions to $\III'$.
Observe that this correspondence is not one-to-one as in general there can be more possibilities how to extend the solution to variables $\hat{x}_{C'}$.
We say that all these solutions project to the same solution $\vex$ to instance $\III$.
However, as all the $C'$\hy variables do not occur in the objective function the value of the objective function of all  solutions that project to $\vex$ is the same.
\end{proof}
By combining the two claims it is possible to transform ILP instance $\III$ to $\hat{\III}$ with the following properties.
\begin{itemize}
  \item all components of $\hat{\III}$ are either extended or for their type $T$ it holds that $\mathrm{mult}_{\III}(T) = N$,
  \item for each type $\hat{T}$ of $\hat{\III}$ it holds that $\mathrm{mult}_{\hat{\III}}(\hat{T}) = N$,
  \item $b_\ell(\hat{\III}) = b_\ell(\III)$,
  \item number of variables in $\hat{\III}$ is at most twice the number of variables in $\III$. \qedhere
\end{itemize}
\end{proof}

The algorithmic consequences of Theorem~\ref{thm:nonuniform4blockFPT} together with Theorems~\ref{thm:nfold4blockip},~\ref{thm:det-fpt} and~\ref{thm:nfoldand2stocip} are the following corollaries.
\begin{COR}\label{cor:parmixedXP}
  Let $a$ and $z$ be constants and let $\III$ be an ILP instance with
  $\coA(\III) \leq a$ and $\fr(\III) \leq z$, then $\III$ can be solved in
  polynomial time.
\end{COR}

\begin{COR}\label{cor:parvarconFPT}
  ILP is fpt when parameterized by $\max\{\coA,\fr^V\}$ and also when
  parameterized by $\max\{\coA,\fr^C\}$.
\end{COR}


\section{Unary ILP}
Here we will prove that \textsc{Unary ILP} is polynomial-time solvable
when $\fr^C$ is bounded by a constant; this contrasts
the case of general ILP, which remains NP-hard in this case (see
Theorem~\ref{lem:binaryconstbdhard} later).
In particular, we will give
an \XP{} algorithm for \textsc{Unary ILP} parameterized by $\fr^C$. We
will also present lower bounds showing that such an algorithm cannot
exist for \textsc{Unary ILP} parameterized by $\fr^V$ or $\fr$, and
rule out the existence of a fixed-parameter algorithm for $\fr^C$.


\subsection{The Algorithm}

The crucial, and also most technically demanding, part of this
result is showing that it suffices to restrict our search
space to assignments over polynomially bounded variable domains.

Before showing this we need some preparation.

\begin{PRO}\label{pro:crammerdet}
 Let $\mathbf{A}$ be an integer $k\times k$ non-singular matrix and
 $\veb$ an integer vector. Then $|\vex[i]|\le k!\coB(\coA)^{k-1}$ for
 the unique $\vex$ such that $\mathbf{A}\vex=\veb$.
\end{PRO}
\begin{proof}
  Because of Proposition~\ref{pro:crammer} it holds that
  $$\vex[i]=\frac{\det(\mathbf{A}(i))}{\det(\mathbf{A})}.$$ Moreover,
  since $\mathbf{A}$ is a non-singular integer matrix, we have that
  $|\det(\mathbf{A})|\geq 1$ and thus \mbox{$|\vex[i]|\leq
  |\det(\mathbf{A}(i))|$}, which together with Proposition~\ref{pro:det} implies
  \mbox{$|\vex[i]|\leq |\det(\mathbf{A}(i))|\leq k!\coB(\coA)^{k-1}$}, as required.
\end{proof}

\begin{LEM}\label{lem:bound-var}
  Let $\mathbf{Q}$ be a $k\times n$ matrix of rank $k$, $\vey$ be
  a vector of $n$ variables, $\ved$ be
  a vector of size $k$, $I$ be a set of $k$ linearly independent columns of
  $\mathbf{Q}$, $V$ be their corresponding variables in $\var(\vey)$,
  and let $\beta$ be an assignment of the variables in
  $\var(\vey)$ such that $\mathbf{Q}\beta(\vey)=\ved$.
  Then for every $v \in V,$ it holds that

  \[
    |\beta(v)| \leq k!\Bigl(\coD+\coQ\sum_{u \in \var(\vey) \setminus
      V}\beta(u)\Bigr)\bigl(\coQ\bigr)^{k-1}.
  \]
\end{LEM}
\begin{proof}
  Let $\vey'$ be $\vey$ restricted to the variables in
  $\var(\vey)\setminus V$ and let $J$ be the set of all columns of $\mathbf{Q}$
  that are not in $I$.
  We will now apply the assignment $\beta$ for the variables in $\vey'$ to $\mathbf{Q}$.
  This will give us a set of equations that
  need to be satisfied for the variables in $V$ allowing us to
  obtain a bound on $\beta$ for these variables. Namely, the
  right-hand side denoted by $\ved'$ of our equations is obtained from $\ved$ by
  subtracting the application of $\beta$ to
  $\mathbf{Q}_{(*,J)}$, i.e.,
  $\ved'=\ved-\mathbf{Q}_{(*,J)}\beta(\vey')$, which after
  restricting $\mathbf{Q}$ to the columns $I$ and using the
  restriction $\vey''$ of $\vey$ to
  the variables in $V$ gives us the
  following equations that are satisfied by $\beta$:

  \begin{equation}\label{eqn:assignC}
    \mathbf{Q}_{(*,I)}\beta(\vey'')=\ved'
  \end{equation}

  Note that because $I$ is a set of $k$ linearly independent columns
  the matrix $\mathbf{Q}_{(*,I)}$ is non-singular. Moreover, observe that
  $\ved'[i] \leq \coD+\coQ\sum_{u \in \var(\vey) \setminus
      V}\beta(u)$ for every $i$ with ${1 \leq i \leq k}$.
  Because $\beta$ satisfies~\ref{eqn:assignC} we obtain from
  Proposition~\ref{pro:crammerdet} that

  \[
    |\beta(v)| \leq k!\Bigl(\coD+\coQ\sum_{u \in \var(\vey) \setminus
      V}\beta(u)\Bigr)\bigl(\coQ\bigr)^{k-1},
  \]
  for every variable $v \in V$.
\end{proof}

The following lemma provides an important ingredient for
Lemma~\ref{lem:boundedconstbd} below. Its proof crucially makes use of
the specific structure of our ILP instance.

\begin{LEM}
\label{lem:k2columns}
  Let $\III$ be an instance of \textsc{Unary ILP} with matrix
  $\mathbf{A}$. Then for any set $D$ of linearly dependent columns of
  $\mathbf{A}$, it holds that $\mathbf{A}_{(*,D)}$ contains a subset of at
  most \mbox{$\fr^C(\III)(\fr^C(\III)+1)$} linearly dependent columns.
\end{LEM}
\begin{proof}
  Let $Z \subseteq
  \cF(\III)$ be a constraint backdoor for $\III$ of size at most
  $\fr^C(\III)$ and let $\ves$ be a non-zero vector satisfying
  $\mathbf{A}_{(*,D)}\ves = \ve 0$.
  Let $C_1,\dotsc,C_p$ be all components of $\III\setminus Z$ that
  contain at least one variable corresponding to a column in $D$ and
  let $D_i$ be the set of all columns in $D$ that correspond to
  variables in $C_i$. Moreover, let $\ves_{C_i}$ be the restriction of
  $\ves$ to the entries corresponding to variables in $C_i$.
  Note that if $p \leq
  \fr^C(\III)+1$, then $D$ already contains at most
  $\fr^C(\III)(\fr^C(\III)+1)$ linearly dependent columns and the
  lemma follows. So we can assume in the following that
  $p>\fr^C(\III)+1$.
  Denote by $\vew_{C_i}$ the vector $\mathbf{A}_{(*,D_i)}\ves_{C_i}$. If
  $\vew_{C_i} = \ve 0$, then the variables in $C_i$ that $\ves$ does not
  assign to $0$ correspond to at most $\fr^C(\III)$ linearly dependent
  columns and the lemma follows.
  Otherwise, it is easy to observe that
  if $\vew_{C_i}[j]\neq 0$ then $j$ corresponds to a constraint in
  $Z$. Hence for every $C_i$ all
  non-zero entries of the vector $\vew_{C_i}$ correspond to constraints in
  $Z$. Consequently any subset of $\fr^C(\III)+1$ vectors from
  $\vew_{C_1},\dotsc,\vew_{C_p}$ in particular the vectors $\vew_{C_1},\dots,
  \vew_{C_{\fr^C(\III)+1}}$ are linearly dependent (since all their
  non-zero entries correspond to constraints in $Z$ and $|Z|\leq \fr^C(\III)$),
  which implies that the set $\bigcup_{1 \leq i \leq
    \fr^C(\III)+1}D_i$ is the required subset of at most
  $\fr^C(\III)(\fr^C(\III)+1)$ linearly dependent columns of $\mathbf{A}_{(*,D)}$.
\end{proof}

\begin{lemma}\label{lem:sol-plus-minus}
  Let $\III=(\mathbf{A}, \vex, \veb, \ver, \veu,\eval)$ be an ILP
  instance, $\alpha$ a solution for
  $\III$, and $\delta$ a non-zero integer vector such that $\alpha+\delta$ and
  $\alpha-\delta$ are feasible assignments for $\III$. Then
  $\eval\delta=0$ and moreover $\alpha+\delta$ and $\alpha-\delta$ are
  also solutions for $\III$.
\end{lemma}
\begin{proof}
  Assume for a contradiction that $\eval\delta\neq 0$, then either
  $\eval(\alpha+\delta)>\eval(\alpha)$ or
  $\eval(\alpha-\delta)>\eval(\alpha)$, contradicting that $\alpha$ is
  a solution.
\end{proof}

We are now ready to show that we only need to consider solutions with
polynomially bounded variable domain.
\begin{lemma}
\label{lem:boundedconstbd}
Let $\III$ be a feasible instance of \textsc{Unary ILP-Feasibility} of size $n$.
Then, there exists a solution $\alpha$ with
$|\alpha(v)| \leq \mb$ for every $v \in \var(\III)$, where
$\mb=8\bigl(2(\fr^C(\III)+2)^2\bigr)!(n)^{2(\fr^C(\III)+2)^2}$.

\end{lemma}
\begin{proof}
  Let $\III=(\mathbf{A}, \vex, \veb, \ver, \veu,\eval)$ be the
  provided instance of \textsc{Unary ILP} and let $Z \subseteq
  \cF(\III)$ be a constraint backdoor witnessing $\fr^C(\III)$.

  Let $\ms=\bigl((\fr^C(\III)+1)^2\bigr)!(n)^{(\fr^C(\III)+1)^2}$
  and $\mm=4\bigl((\fr^C(\III)+2)^2\bigr)!(n)^{(\fr^C(\III)+2)^2}$.
  For a solution $\alpha$ of $\III$, let
  $V(\alpha)$ be the set of all variables $v$ of $\III$ such
  that $|\alpha(v)| \geq 2\ms$.
  Let us now consider a solution
  $\alpha$ which minimizes the size of $V(\alpha)$. Observe that
  because $\mb \geq 2\ms$ it holds that if
  $|V(\alpha)|=0$ then the lemma holds, and so we may assume that
  $V(\alpha)$ is non-empty.


  In the following we consider the submatrix $\mathbf{B}=\mathbf{A}_{(*,V(\alpha))}$.
  Let us first consider the case where the
  columns of $\mathbf{B}$ are linearly dependent. We show that in this
  case, we can find a solution $\alpha'$ such that $|V(\alpha')|
  < |V(\alpha)|$, which contradicts the choice of
  $\alpha$.

  Because of Lemma~\ref{lem:k2columns} there is a non-empty set $O$ of
  linearly dependent columns of $\mathbf{B}$ of size at most $\fr^C(\III)(\fr^C(\III)+1)$.
  Consider a subset $Y = \{\vev^1, \dots, \vev^{|Y|}\}$ of linearly
  dependent columns of $O$ such
  that the columns of each proper subset of $Y$ are
  linearly independent and let $X=Y \setminus \{\vev^{|Y|}\}$.
  Because $Y$ is a minimal set of linearly dependent columns, it holds that
  there is a vector $\vea$ without any zero entries such that
  $\mathbf{B}_{(*,Y)}\vea=\ve 0$, which implies the existence of a
  vector $\vea_X$, again without zero entries, such that
  $\mathbf{B}_{(*,X)}\vea_{X}=\vev^{|Y|}$. We will show that
  there is such a vector $\vea$ that is integer and satisfies $|\vea[i]|\le \ms$ for every
  $1 \leq i \leq |Y|$.
  We start by solving $\mathbf{B}_{(*,X)}\vea_{X}=\vev^{|Y|}$ using
  Cramer's rule. Because the columns in $X$ are linearly
  independent, it follows that $\mathbf{B}_{(*,X)}$ has a set $R$ of
  linearly independent rows with $|R|=|X|$. Then because
  the matrix $\mathbf{B}_{(R,X)}$ is non-singular, we have that there
  is a unique $\vea_{X}$ such that
  $\mathbf{B}_{(R,X)}\vea_{X}=\vev^{|Y|}_{R}$, where $\vev^{|Y|}_{R}$
  denotes the restriction of the vector $\vev^{|Y|}$ to the
  entries associated with the columns in $R$.
  Moreover, because there is a non-zero vector $\vea_{X}$ with
  $\mathbf{B}_{(*,X)}\vea_{X}=\vev^{|Y|}$, it follows
  that the unique vector $\vea_{X}$ satisfying
  $\mathbf{B}_{(R,X)}\vea_{X}=\vev^{|Y|}_{R}$ also satisfies
  $\mathbf{B}_{(*,X)}\vea_{X}=\vev^{|Y|}$. Using Cramer's Rule, we obtain
  $\vea_{X}[i]=\frac{\det(\mathbf{B}_{(R,X)}(i))}{\det(
    \mathbf{B}_{(R,X)})}$ for every $i$ with $1 \leq i \leq |X|$ as
  the unique vector satisfying
  $\mathbf{B}_{(R,X)}\vea_{X}=\vev^{|Y|}_{R}$.

  Hence the vector $\ved$ with
  $\ved[i]=\vea_X[i]\det(\mathbf{B}_{(R,X)})=\det(\mathbf{B}_{(R,X)}(i))$
  for every $i$ with $1\leq i \leq |X|$ and
  $\ved[|Y|]=-\det(\mathbf{B}_{(R,X)})$ is a non-zero integer vector that satisfies
  $\mathbf{B}_{(*,Y)}\ved=\ve 0$. From Proposition~\ref{pro:det},
  we obtain that
  \begin{equation*}
    |\ved[i]|
    \leq \bigl(\fr^C(\III)(\fr^C(\III)+1)\bigr)!(\coA)^{\fr^C(\III)(\fr^C(\III)+1)}
    \leq \bigr((\fr^C(\III)+1)^2\bigl)!(n)^{(\fr^C(\III)+1)^2}
    =\ms,
  \end{equation*}
  as required.

  For notational convenience we will in the following assume that
  $\mathbf{A}$ starts with the columns $\vev^1,\dotsc,\vev^{|Y|}$ from
  $Y$. Let $\vew$ be the vector defined by:
  \begin{itemize}
  \item $\vew[i]=\ved[i]$, if $i \leq |Y|$, and
  \item $\vew[i]=0$ otherwise
  \end{itemize}
  Note that $\mathbf{A}\vew=\ve 0$.
  For an integer $\Delta$, let $\alpha_\Delta: \var(\III) \rightarrow \mathbb{Z}$
  denote the assignment $\alpha_\Delta=\alpha+\Delta\vew$.
  Note that $\alpha_\Delta$ is an integral assignment, moreover
  because
  \[
    \mathbf{A}\alpha_\Delta(\vex)=\mathbf{A}\alpha(\vex)+\Delta\mathbf{A}\vew=\mathbf{A}\alpha(\vex)
  \]
  it follows that $\alpha_\Delta$ is a feasible integral assignment for
  $\mathbf{A}\vex=\veb$ for every $\Delta \in \mathbb{Z}$.
  Let $\Delta$ be the integer with smallest absolute value such that
  there is at least one variable $v \in V(\alpha)$ with
  $|\alpha_\Delta(v)| \leq 2\ms$. We claim that for every
  $|\delta|\leq |\Delta|$, $\alpha_\delta$ is a
  solution for $\III$. We first show that $\ver[i] \leq
  \alpha_\delta(\vex[i])\leq \veu[i]$ for every $i$ with $1 \leq i
  \leq |\var(\III)|$. If $\vex[i]$ corresponds to a column that is not
  in $Y$, then $\alpha_\delta(\vex[i])=\alpha(\vex[i])$, which implies
  $\ver[i] \leq \alpha_\delta(\vex[i])\leq \veu[i]$. Otherwise,
  assume w.l.o.g. that $\alpha(\vex[i])\geq 0$ (the case that
  $\alpha(\vex[i])< 0$ is symmetric).
  Because $\alpha(\vex[i])\geq 2\ms$ and $|\ved[j]| \leq \ms$ for
  every $j$ with $1 \leq j \leq |Y|$ together with
  the choice of $\Delta$, we obtain that
  $\ms \leq \alpha_\delta(\vex[i])$.
  Because $\ms> n$ and since $\alpha$ is a feasible solution it
  follows that $\veu[i]=\infty$ and $\ver[i]\leq \ms$, which shows
  that $\ver[i] \leq \alpha_\delta(\vex[i])\leq \veu[i]$.
  Hence in particular $\alpha_\Delta$ and also $\alpha_1=\alpha+\vew$
  and $\alpha_{-1}=\alpha-\vew$ are feasible
  assignments, which together with Lemma~\ref{lem:sol-plus-minus}
  (after setting $\delta$ to $\vew$) implies that $\eval\vew=0$ and
  hence $\eval\alpha=\eval\alpha_\Delta$. Consequently $\alpha_\Delta$ is a
  solution for $\III$ with $|V(\alpha_\Delta)|<|V(\alpha)|$, contradicting our choice of $\alpha$.

  We conclude that the columns of $\mathbf{B}$ must be linearly
  independent, which implies that
  there is a set $R$ of $|V(\alpha)|$ linearly independent rows in
  $\mathbf{B}$. Consider the set $S$ of all components of $\III \setminus Z$
  that have a non-empty intersection with either $V(\alpha)$ or the
  constraints corresponding to the rows in $R$.
  Let $C_1,\dotsc,C_p$ be the restrictions of the components in $S$
  to the variables in $V(\alpha)$ and
  the constraints in $R$.

  Observe that for every component $C_i$, it holds that the rows in
  $R$ that correspond to constraints in $C_i$ are zero everywhere but at the
  entries corresponding to variables in $C_i$. Because the rows in $R$
  are independent it follows that every component must have at least
  as many variables as constraints.
  Moreover, because $\mathbf{B}_{(R,*)}$ is a square matrix and the only
  rows in $R$ that do not correspond to constraints in components, correspond to
  the constraints in $Z$, we obtain that there are at most $|Z|\leq \fr^C(\III)$
  components that have strictly more variables than constraints,
  all other components have the same number of rows and columns.
  Let $C_i$ be a component with the same number of rows as columns and
  let $C_i'$ be the unique component of $\III\setminus Z$
  containing $C_i$. Let $Q=\mathbf{A}_{(C(C_i'),R(C_i))}$ and $\vey$ be the
  subvector of $\vex$ restricted to the variables of $C_i'$,
  $\ved$ be the subvector of $\veb$ restricted to entries that correspond
  to the constraints of $C_i$, $V=\var(C_i)$, $I$ the set
  of columns of $Q$ corresponding to the variables in $V$,
  and $\beta$ the assignment $\alpha$ restricted to
  the variables in $\vey$. Because the rows in $Q$ are independent its
  rank is $|\cF(C_i)|$, because $\alpha$ satisfies
  $\mathbf{A}\alpha(\vex)=\veb$ and all but the columns corresponding to
  the variables in $\var(C_i')$ of $\mathbf{A}_{*,\cF(C_i)}$ are zero
  everywhere, it holds that $\mathbf{Q}\beta(\vey)=\ved$.
  Hence we can apply Lemma~\ref{lem:bound-var} for $\mathbf{Q}$, $\vey$, $\ved$,
  $V$, $I$, and $\beta$ and obtain:
  \begin{align*}
    |\alpha(v)| &\leq \fr^C(\III)!\Bigl(\coB+\coA\sum\limits_{u \in \var(C_i') \setminus
      \var(C_i)}\alpha(u)\Bigl)\bigr(\coA\bigl)^{\fr^C(\III)-1} \\
     &\leq \fr^C(\III)!\bigl(\coB+\coA\fr^C(\III)2\ms\bigr)(\coA)^{\fr^C(\III)-1} \\
     &\leq \fr^C(\III)!4\ms\fr^C(\III)(n)^{\fr^C(\III)}
     \leq 4\bigl((\fr^C(\III)+2)^2\bigr)!(n)^{(\fr^C(\III)+2)^2} \\
     &= \mm
  \end{align*}
  for every variable $v \in V$.
  The second to last inequality follows because $|\alpha(v)|\leq 2\ms$
  for every $v$ in $\var(C_i')\setminus \var(C_i)$, which is because
  $(\var(C_i')\setminus \var(C_i)) \subseteq (\var(\III) \setminus V(\alpha))$.
  This shows that the assignment $\alpha$ is bounded by $\mm$ for all
  variables contained in components $C_i$ that have the same number
  of variables and constraints. Consider the remaining components
  $D_1,\dotsc,D_s$ among $C_1,\dotsc,C_p$, i.e., the components among
  $C_1,\dotsc,C_p$ that have more variables than constraints. Recall
  that $s \leq |Z|\leq \fr^C(\III)$. Let $V=\bigcup_{1 \leq i \leq
    s}\var(D_i)$ and let $J$ be the corresponding columns of $V$ in
  $\mathbf{A}$. Note that $|J|\leq (\fr^C(\III))^2$.
  Because $V \subseteq V(\alpha)$ it holds
  that $J$ is a set of linearly independent columns. Hence there is a
  set $R'$ of $|J|$ linearly independent rows in $\mathbf{A}_{(*,J)}$.

  Let $\mathbf{Q}=\mathbf{A}_{(R',*)}$, $\vey=\vex$,
  $\ved$ be the subvector of $\veb$ restricted to entries that correspond
  to the rows in $R'$, $I$ be the columns in $J$ restricted to the
  rows in $R'$, and $\beta=\alpha$. Because the rows in $\mathbf{Q}$ are independent its
  rank is $|I|$, because $\mathbf{Q}$ is a submatrix of $\mathbf{A}$
  only restricted in rows, we have $\mathbf{Q}\beta(\vey)=\ved$.
  Hence we can apply Lemma~\ref{lem:bound-var} for $\mathbf{Q}$, $\vey$, $\ved$,
  $V$, $I$, and $\beta$ and obtain:
  \begin{align*}
    |\alpha(v)|
    & \leq \bigl(\fr^C(\III)^2\bigr)!\Bigl(\coB+\coA\sum\limits_{u \in \var(\III) \setminus V}\alpha(u)\Bigr)\bigl(\coA\bigr)^{(\fr^C(\III))^2-1)} \\
    & \leq \bigl(\fr^C(\III)^2\bigr)!\bigl(\coB+\coA|\var(\III)|\mm\bigr)(\coA)^{(\fr^C(\III))^2-1)} \\
    & \leq 8\bigl(2(\fr^C(\III)+2)^2\bigr)!(n)^{2(\fr^C(\III)+2)^2} \\
    & = \mb
  \end{align*}
  for every variable $v \in V$. The second to last inequality follows
  follows because $|\alpha(v)|\leq \mm$ for every $v$ in
  $\var(\III)\setminus V$, as shown previously.
  This concludes the proof of the lemma.
\end{proof}

To complete the proof of the desired statement, we use a recent
result of~\cite[Proposition 2 and Theorem 11]{GanianOrdyniakRamanujan17} on solving ILP using
treewidth (which is always at most $\fr$) and obtain:

\begin{PRO}[Proposition 2 and Theorem 11 in \citeauthor{GanianOrdyniakRamanujan17}
  \citeyear{GanianOrdyniakRamanujan17}]\label{pro:incidencetw}
  Let $\III=(\mathbf{A}, \vex, \veb, \ver, \veu,\eval)$ be an ILP
  with incidence treewidth $\omega$ and such that
  $\ver[i]\neq -\infty$ and $\veu[i] \neq \infty$ for every entry
  $i$.
  Then $\III$ can be solved in time
  $\bigoh((\coA \cdot \Delta \cdot
  |var(\III)|)^{\omega})(|\var(\III)|+|\cF(\III)|)$, where
  $\Delta=\max_i\{|\ver[i]|,|\veu[i]\}$.
\end{PRO}

\begin{theorem}
\label{thm:unaryconXP}
\textsc{Unary ILP} is polynomial-time solvable for any fixed value of $\fr^C(\III)$, where $\III$ is the input instance.
\end{theorem}
\begin{proof}
  Let $\III$ be an input instance of \textsc{Unary ILP} encoded in $n$
  bits and let $\III'$ be the instance obtained from $\III$ by
  replacing $-\infty$ and $\infty$ entries in $\ver$ and $\veu$ with
  $-\mb$ and $\mb$, respectively (for the definition of $\mb$ see the
  statement of Lemma~\ref{lem:boundedconstbd}). It follows from
  Lemma~\ref{lem:boundedconstbd} that $\III$ and $\III'$ are
  equivalent ILP instances. Now let $\omega$ be the incidence
  treewidth of $\III'$ (which is equal to the incidence treewidth of
  $\III$).
  Observe that $\omega\leq \fr^C(\III)$ and
  hence it follows from Proposition~\ref{pro:incidencetw} that $\III'$
  (and thus also $\III$) can be solved in time
  $\bigoh((\coA \cdot \mb \cdot |var(\III)|)^{\fr^C(\III)})(|\var(\III)|+|\cF(\III)|)$.
\end{proof}

\subsection{Lower Bounds}
We complement our algorithm with matching lower bounds:
strong \NP{}\hy hardness for variable and mixed backdoors, \W{1}\hy hardness
in the case of constraint backdoors, and weak \NP{}\hy hardness for
constraint and mixed backdoors.

\begin{theorem}
\label{thm:unaryvarhard}
  \textsc{Unary ILP-feasibility} is \paraNP{}\hy hard parameterized by $\fr^V(\III)$.
\end{theorem}
\begin{proof}
  We prove the theorem by a polynomial-time reduction from the
  well-known \NP{}\hy hard \textsc{$3$-Colorability}
  problem~\cite{GareyJohnson79}: given a graph, decide whether the vertices of $G$
  can be colored with three colors such that no two adjacent vertices
  of $G$ share the same color.

  The main idea behind the reduction is
  to represent a $3$-partition of the vertex set of $G$ by the domain values of
  three ``global'' variables. The value of each of these global variables
  will represent a subset of vertices of $G$ that will be colored using
  the same color. To represent a subset of the vertices of $G$ in
  terms of domain values of the global variables, we will associate
  every vertex of $G$ with a unique prime number and represent a subset by the
  value obtained from the multiplication of all prime numbers of
  vertices contained in the subset. To ensure that the subsets
  represented by the global variables correspond to a valid
  $3$-partition of $G$ we will introduce constraints which ensure
  that:
  \begin{itemize}
  \item[C1]  For every prime number representing some vertex of $G$
    exactly one of the global variables is divisible by that prime
    number. This ensures that every vertex of $G$ is assigned to
    exactly one color class.
  \item[C2] For every edge $\{u,v\}$ of $G$ it holds that no global
    variable is divisible by the prime numbers
    representing $u$ and $v$ at the same time. This ensures that no
    two adjacent vertices of $G$ are assigned to the same color class.
  \end{itemize}
  Thus let $G$ be the given instance of \textsc{$3$-Coloring} and
  assume that the vertices of $G$ are uniquely identified as elements of $\{1,\dots,|V(G)|\}$.
  In the following we denote by $p(i)$ the
  $i$-th prime number for any positive integer $i$, where $p(1)=2$.
  We construct an instance $\III$ of
  \textsc{ILP-feasibility} in polynomial time with $\fr^V(\III) \leq 25$,
  and coefficients bounded by a polynomial in $V(G)$ such that $G$ has a
  $3$\hy coloring if and only if $\III$ has a feasible assignment.
  This instance $\III$ has the following variables:
  \begin{itemize}
  \item The \emph{global variables} $c_1$, $c_2$, and $c_3$ with
    an arbitrary positive domain, whose
    values will represent a valid $3$-Partitioning of $V(G)$.
  \item For every $i$ and $j$ with $1 \leq i \leq |V(G)|$ and $1 \leq
    j \leq 3$, the variables $m_{i,j}$, $\textup{sl}^1_{i,j}$, and
    $\textup{sl}^2_{i,j}$ (with an arbitrary non-negative domain),
    $r_{i,j}$ (with domain between $0$ and $p(i)-1$), and
    $u_{i,j}$ (with binary domain). These variables are used to secure
    condition C1.
  \item For every $e \in E(G)$, $i \in e$, and $j$ with $1 \leq
    j \leq 3$, the variables $m_{e,i,j}$, $\textup{sl}^3_{e,i,j}$,
    $\textup{sl}^4_{e,i,j}$, and $\textup{sl}^5_{e,j}$ (with an arbitrary non-negative domain),
    $r_{e,i,j}$ (with domain between $0$ and $p(i)-1$), and
    $u_{e,i,j}$ (with binary domain). These variables are used to secure condition C2.
  \end{itemize}
  Note that the variables $\textup{sl}^1_{i,j}$,
  $\textup{sl}^2_{i,j}$, $\textup{sl}^3_{e,i,j}$,
  $\textup{sl}^4_{e,i,j}$, and $\textup{sl}^5_{e,i}$ are so-called
  ``Slack'' variables, whose sole purpose is to obtain an ILP instance
  that is in equation normal form.
  The instance $\III$ has the following constraints (in the following let $\alpha$ be
  any feasible assignment of $\III$):
  \begin{itemize}
  \item domain restrictions for all variables as given above, i.e.:
    \begin{itemize}
    \item  for every $i$ and $j$ with $1 \leq i \leq
      |V(G)|$ and $1 \leq j \leq 3$, the constraints $c_j \geq 0$,
      $m_{i,j} \geq 0$, $\textup{sl}^1_{i,j} \geq 0$,
      $\textup{sl}^2_{i,j} \geq 0$, $0 \leq r_{i,j}\leq p(i)-1$, and $0\leq u_{i,j}
      \leq 1$.
    \item for every $e\in E(G)$, $i \in e$, and $j$ with $1 \leq j
      \leq 3$, the constraints $m_{e,i,j}\geq 0$,
      $\textup{sl}^3_{e,i,j} \geq 0$, $\textup{sl}^4_{e,i,j} \geq 0$, $\textup{sl}^5_{e,j} \geq 0$, $0\leq r_{e,i,j}\leq
      p(i)-1$, and $0 \leq u_{e,i,j}\leq 1$.
    \end{itemize}
  \item The following constraints, introduced for each $1 \leq i
      \leq |V(G)|$ and $1 \leq j \leq 3$, together guarantee that condition C1 holds:
    \begin{itemize}
    \item Constraints that ensure that $\alpha(r_{i,j})$
      is equal to the remainder of $\alpha(c_j)$ divided by $p(i)$, i.e., the constraint
      $c_j=p(i)m_{i,j}+r_{i,j}$.
    \item
      Constraints that ensure that $\alpha(u_{i,j})=0$
      if and only if $\alpha(r_{i,j})=0$, i.e., the
      constraints $u_{i,j} + \textup{sl}^1_{i,j} = r_{i,j}$ and $r_{i,j} + \textup{sl}^2_{i,j} =
      (p(i)-1)u_{i,j}$.
      Note that together the above constraints now ensure that
      $\alpha(u_{i,j})=0$ if and only if $c_j$ is divisible by $p(i)$.
    \item
      Constraints that ensure that exactly one of $\alpha(u_{i,1})$,
      $\alpha(u_{i,2})$, and $\alpha(u_{i,3})$ is equal to $0$, i.e.,
      the constraints $u_{i,1}+u_{i,2}+u_{i,3} = 2$.
      Note that together all the above constraints now ensure condition
      C1 holds.
    \end{itemize}
  \item The following constraints, introduced for each $1 \leq j \leq 3$, together guarantee that condition C2 holds:
    \begin{itemize}
    \item Constraints that ensure that for every $e \in E(G)$ and $i \in e$,
      it holds that $\alpha(r_{e,i,j})$
      is equal to the remainder of $c_j$ divided by $p(i)$, i.e., the constraint
      $c_j=p(i)m_{e,i,j}+r_{e,i,j}$.
    \item
      Constraints that ensure that for every $e \in E(G)$, $i \in e$,
      and $j$ with $1 \leq j \leq 3$ it holds that
      $\alpha(u_{e,i,j})=0$ if and only if $\alpha(r_{e,i,j})=0$, i.e., the
      constraints $u_{e,i,j} + \textup{sl}^3_{e,i,j} = r_{e,i,j}$ and
      $r_{e,i,j} + \textup{sl}^4_{e,i,j} =
      (p(i) - 1) u_{e,i,j}$.
      Note that together the above constraints now ensure that
      $\alpha(u_{e,i,j})=0$ if and only if $c_j$ is divisible by $p(i)$.
    \item
      Constraints that ensure that for every $e=\{i,k\} \in E(G)$
      and $j$ with $1 \leq j \leq 3$ it holds that at least one of
      $\alpha(u_{e,i,j})$ and $\alpha(u_{e,k,j})$ is non-zero, i.e., the
      constraint $u_{e,i,j}+u_{e,k,j}-\textup{sl}^5_{e,j}= 1$. Note that together with
      all of the above constraints this now ensures condition C2.
    \end{itemize}
  \end{itemize}
  This completes the construction of $\III$. Clearly $\III$ can be
  constructed in polynomial time, and the largest
  coefficient used by $\III$ is equal to $p(|V(G)|)$. It is well-known that $p(i)$ is upper-bounded by
  $O(i\log i)$ due to the Prime Number Theorem, and so
  this in particular implies that the numbers which occur in $\III$ are
  bounded by a polynomial in $|V(G)|$.

  Following the construction and explanations provided above, it is not difficult to see that $\III$ has a
  feasible assignment if and only if $G$ has a $3$\hy coloring. Indeed, for any $3$\hy coloring of $G$, one can construct a feasible assignment of $\III$ by computing the prime-number encoding for vertices that receive colors $1,2,3$ and assign these three numbers to $c_1,c_2,c_3$, respectively. Such an assignment allows us to straightforwardly satisfy the constraints ensuring C1 holds (since each prime occurs in exactly one global constraint), the constraints ensuring C2 holds (since each edge is incident to at most one of each color) while maintaining the domain bounds.

  On the other hand, for any feasible assignment $\alpha$, clearly each of $\alpha(c_1), \alpha(c_2), \alpha(c_3)$ will be divisible by some subset of prime numbers between $2$ and $p(|V(G)|)$. In particular, since $\alpha$ is feasible it follows from the construction of our first group of constraints that each prime between $2$ and $p(|V(G)|)$ divides precisely one of $\alpha(c_1), \alpha(c_2), \alpha(c_3)$, and so this uniquely encodes a corresponding candidate $3$\hy coloring for the vertices of the graph. Finally, since $\alpha$ also satisfies the second group of constraints, this candidate $3$\hy coloring must have the property that each edge is incident to at exactly $2$ colors, and so it is in fact a valid $3$-coloring.


  It remains to show that $\fr^V(\III)\leq 25$. We show this by showing
  that the set $B=\{c_1,c_2,c_3\}$ is a variable backdoor set to
  $25$-compactness. Note that the graph $G_\III \setminus
  \{c_1,\dotsc,c_3\}$ has only two types of components (all other
  components are isomorphic to one of the two types):
  \begin{itemize}
  \item for every $i$ with $1 \leq i \leq |V(G)|$,
    one component containing the variables $m_{i,1},\dotsc, m_{i,3}$,
    $\textup{sl}^1_{i,1},\dotsc, \textup{sl}^1_{i,3}$,
    $\textup{sl}^2_{i,1},\dotsc, \textup{sl}^2_{i,3}$,
    $r_{i,1},\dotsc,r_{i,3}$, $u_{i,1},\dotsc,u_{i,3}$. Moreover, these $15$
    variables occur in exactly $10$ constraints together; these are
    the constraints introduced above to ensure condition C1. Hence the
    total size of these components is $25$.
  \item for every $e=\{w,v\} \in E(G)$ and $j$ with $1 \leq j \leq 3$,
    one component on the vertices $m_{e,w,j}$, $m_{e,v,j}$,
    $\textup{sl}^3_{e,w,j}$, $\textup{sl}^4_{e,v,j}$,
    $r_{e,w,j}$, $r_{e,v,j}$, $u_{e,w,j}$, $u_{e,v,j}$, and
    $\textup{sl}^5_{e,j}$. Moreover, these $9$
    variables occur in exactly $7$ constraints together; these are
    the constraints introduced above to ensure condition C2. Hence the
    total size of these components is $16$.
  \end{itemize}
  This show that $B$ is a variable backdoor to $25$-compactness, as
  required.
\end{proof}

\begin{theorem}
\label{thm:unaryconhard}
  \textsc{Unary ILP-feasibility} is \Weft\emph{[1]}\hy hard parameterized by $\fr^C(\III)$.
\end{theorem}
\begin{proof}
  We prove the theorem by a  parameterized reduction from
  \textsc{Multicolored Clique}, which is well-known to be
  \W{1}\hy complete~\cite{Pietrzak03}.
  Given an integer $k$ and a
  $k$-partite graph $G$ with partition $V_1,\dotsc,V_k$,
  the \textsc{Multicolored Clique} problem
  ask whether $G$ contains a $k$-clique. In the following we denote by
  $E_{i,j}$ the set of all edges in $G$ with one endpoint in $V_i$ and
  the other endpoint in $V_j$, for every $i$ and $j$ with $1\leq i < j
  \leq k$.
  To show the theorem, we will construct an instance $\III$ of
  \textsc{ILP-feasibility} in polynomial time that has a constraint
  backdoor set of size $2k+2\binom{k}{2}$ to $3$-compactness and
  coefficients bounded by a polynomial in $|V(G)|$ such that $G$ has a
  $k$-clique if and only if $\III$ has a feasible assignment.

  The main idea behind the reduction is to first guess one vertex from
  each part $V_i$ and one edge between every two parts $V_i$ and $V_j$
  and to then verify that the selected vertices and edges form
  a $k$-clique in $G$.

  The first step is achieved by introducing one
  binary variable for every vertex and edge of $G$ together with
  $2k+2\binom{k}{2}$ global constraints that ensure that (1) exactly one
  of the variables representing the vertices in $V_i$ is set to one
  and (2) exactly one of the variables representing the edges between
  $V_i$ and $V_j$ is set to one. The
  second step, i.e., verifying that the chosen vertices and edges
  indeed form a $k$-clique of $G$, is achieved by identifying each
  vertex of $G$ with a unique number such that the sum of any two numbers
  assigned to two vertices of $G$ is unique. By identifying each
  edge of $G$ with the sum of the numbers assigned to its endpoints, it
  is then possible to verify that the selected vertices and edges form
  a $k$-clique by checking whether the number assigned to the selected
  edge $e$ is equal to the sum of the numbers assigned to
  the selected vertices in $V_i$ and $V_j$. Sets of
  numbers for which the sum of every two numbers from the set is
  unique are also known as Sidon sequences.
  Indeed a \emph{Sidon
    sequence} is a sequence of natural numbers such that the sum of every two
  distinct numbers in the sequence is unique. For our reduction we
  will need a Sidon sequence of $|V(G)|$ natural numbers, i.e.,
  containing one number for each vertex of $G$. Since the numbers in
  the Sidon sequence will be used as coefficients of $\III$, we need
  to ensure that the largest of these numbers is bounded by a
  polynomial in $G$. Indeed~\cite{ErdosTuran41} shows that a Sidon sequence
  containing $n$ elements and whose largest element is at most $2p^2$,
  where $p$ is the smallest prime number larger or equal to $n$ can be
  constructed in polynomial time. Together with Bertrand's
  postulate~\cite{AignerZiegler10}, which states that for every natural
  number $n$ there is a prime number between $n$ and $2n$, we obtain
  that a Sidon sequence containing $|V(G)|$ numbers and whose largest
  element is at most $8|V(G)|^2$ can be found in polynomial time.
  In
  the following we will
  assume that we are given such a Sidon sequence $\SSS$ and we denote
  by~$\SSS(i)$ the~$i$\hy th element of $\SSS$ for any $i$ with $1 \leq i
  \leq |V(G)|$. Moreover, we denote by $\max(\SSS)$ and $\max_2(\SSS)$
  the largest element of~$\SSS$ respectively the maximum sum of any
  two numbers in~$\SSS$.

  We are now ready to construct the instance $\III$ of
  \textsc{ILP\hy feasibility} such that $G$ has a
  $k$-clique if and only if $\III$ has a feasible assignment.
  This instance $\III$ has the following variables:
  \begin{itemize}
  \item For every $v \in V(G)$ a binary variable $v$ (with domain
    $\{0,1\}$) that is $1$ if $v$ is selected to be in the $k$-clique
    and $0$ otherwise.
  \item For every $e \in E(G)$ a binary variable $e$ (with domain
    $\{0,1\}$) that is $1$ if $e$ is selected to be in the $k$-clique
    and $0$ otherwise.
  \item For every $i$ with $1 \leq i \leq k$, a variable $v_i$ (with
    unrestricted domain), which
    will be set to $\SSS(v)$ if the vertex $v \in V_i$ was
    selected to be in the $k$-clique.
  \item For every $i$ and $j$ with $1 \leq i < j \leq k$, a variable
    $e_{i,j}$ (with unrestricted domain), which will be set to
    $\SSS(v)+\SSS(u)$ if the edge $e \in E_{i,j}$ with $e=\{u,v\}$ was
    selected to be in the $k$-clique.
  \end{itemize}
  \noindent $\III$ has the following constraints:
  \begin{itemize}
  \item Constraints that restrict the domains of all variables as
    specified above, i.e.:
    \begin{itemize}
    \item for every $v \in V(G)$, the constraints $0 \leq v \leq 1$.
    \item for every $e \in E(G)$, the constraints $0 \leq e \leq 1$.
    \end{itemize}
    We will denote by
    $D$ the set of all these constraints.
  \item for every $i$ with $1 \leq i \leq k$, the constraint \mbox{$\sum_{v
      \in V_i}v=1$}, which ensures that from every part $V_i$ exactly
    one vertex is selected to be in the $k$-clique. We will denote by
    $V_{\mathit{SEL}}$ the set of all these constraints.
  \item for every $i$ and $j$ with $1 \leq i < j \leq k$, the constraint $\sum_{e
      \in E_{i,j}}e=1$, which ensures that between any two parts $V_i$
    and $V_j$ exactly one edge is selected to be in the $k$-clique.
    We will denote by $E_{\mathit{SEL}}$ the set of all these constraints.
  \item for every $i$ with $1 \leq i \leq k$, the constraint \mbox{$\sum_{v
      \in V_i}\SSS(v)v=v_i$}, which ensures that $v_i$ is equal to $\SSS(v)$
    whenever $v$ is selected for the $k$-clique.
    We will denote by $V_{\mathit{ASS}}$ the set of all these constraints.
  \item for every $i$ and $j$ with $1 \leq i < j \leq k$, the constraint $\sum_{e=\{u,v\}
      \in E_{i,j}}(\SSS(u)+\SSS(v))e=e_{i,j}$, which ensures that $e_{i,j}$ is equal to $\SSS(u)+\SSS(v)$
    whenever the edge $e \in E_{i,j}$ with endpoints $u$ and $v$ is
    selected for the $k$-clique.
    We will denote by $E_{\mathit{ASS}}$ the set of all these constraints.
  \item for every $i$ and $j$ with $1 \leq i < j \leq k$, the
    constraint $v_i+v_j=e_{i,j}$, which ensures that between any two parts $V_i$
    and $V_j$ the vertices selected for the clique are equal to the
    endpoints of the edge chosen between the two parts.
    We will denote by $\mathit{VE}_{\mathit{CHECK}}$ the set of all these constraints.
  \end{itemize}
  This completes the construction of $\III$. Clearly $\III$ can be
  constructed in polynomial time, and the largest
  coefficient used by $I$ is equal to $\max_2(\SSS)$, which is at most
  ${2\max(\SSS)\leq 16|V(G)|^2}$. We first show that $\III$ has a small
  constraint backdoor to $3$-compactness, and hence our parameter can
  bounded in terms of $k$.
  Namely, we claim that the set $B=V_{\mathit{SEL}} \cup E_{\mathit{SEL}} \cup V_{\mathit{ASS}} \cup
  E_{\mathit{ASS}} \cup \mathit{VE}_{\mathit{CHECK}}$ of constraints of $\III$ is a constraint
  backdoor of size at most $2k+3\binom{k}{2}$ to $3$-compactness. Clearly,
  the components of $G_\III \setminus B$ have size at most $3$, i.e.,
  $G_\III$ has one component of size one for every variable in
  $\{v_1,\dotsc,v_k,e_{1,2},\dotsc,e_{k-1,k}\}$ as well as one
  component of size $3$ for every $a \in V(G) \cup E(G)$, containing the variable
  $a$ together with the two constrains $0 \leq a$ and $a \leq 1$.
  Moreover, the sets $V_{\mathit{SEL}}$, $E_{\mathit{SEL}}$, $V_{\mathit{ASS}}$, $E_{\mathit{ASS}}$, and
  $\mathit{VE}_{\mathit{CHECK}}$ have sizes at most $k$, $\binom{k}{2}$, $k$,
  $\binom{k}{2}$, and $\binom{k}{2}$, respectively, which implies that
  $|B|\leq 2k+3\binom{k}{2}$.

  It remains to show that $G$ has $k$-clique if and only if $\III$ is
  feasible. For the forward direction suppose that $G$ has a
  $k$-clique on the vertices $c_1,\dotsc,c_k$, where $c_i \in V_i$ for
  every $i$ with $1 \leq i \leq k$. Then it is straightforward to
  verify that the assignment $\alpha$ with:
  \begin{itemize}
  \item $\alpha(c_i)=1$ for every $i$ with $1 \leq i \leq k$ and
    $\alpha(v)=0$ for every $v \in V(G) \setminus \{c_1,\dotsc,c_k\}$,
  \item $\alpha(\{c_i,c_j\})=1$ for every $i$ and $j$ with ${1 \leq i<j \leq k}$ and
    $\alpha(e)=0$ for every \mbox{$e \in E(G) \setminus \SB \{c_i,c_j\} \SM
    1 \leq i < j \leq k\SE$},
  \item $\alpha(v_i)=\SSS(c_i)$ for every $i$ with $1 \leq i \leq k$, and
  \item $\alpha(e_{i,j})=\SSS(c_i)+\SSS(c_j)$ for every $i$ and $j$ with ${1
    \leq i < j \leq k}$
  \end{itemize}
  is a feasible assignment for $\III$.

  For the reverse direction suppose that we are given a feasible
  assignment $\alpha$ for $\III$. Then because $\alpha$ satisfies the
  constraints in $D \cup V_{\mathit{SEL}}\cup E_{\mathit{SEL}}$ we obtain that for every
  $i$ and $j$ with $1 \leq i < j \leq k$ it holds that exactly one of
  the variables in $V_i$ and exactly one of the variables in $E_{i,j}$
  is set to one. Let $c_i$ denote the unique vertex in $V_i$ with
  $\alpha(c_i)=1$ and similarly let $d_{i,j}$ denote the unique edge
  in $E_{i,j}$ with $\alpha(d_{i,j})=1$.
  It follows from the constraints in $V_{\mathit{ASS}}$ that
  $\alpha(v_i)=\SSS(c_i)$ and similarly using the constraints in
  $E_{\mathit{ASS}}$ we obtain that $\alpha(e_{i,j})=\SSS(u)+\SSS(v)$, where $u$ and
  $v$ are the endpoints of the edge $d_{i,j}$ in $G$. Moreover, we
  obtain from the constraints in $\mathit{VE}_{\mathit{CHECK}}$ that $v_i+v_j=e_{i,j}$
  and hence $\SSS(c_i)+\SSS(c_j)=\SSS(u)+\SSS(v)$, where again $u$ and $v$ are the
  endpoints of the edge $d_{i,j}$ in $G$. Because $\SSS$ is a Sidon
  sequence, it follows that this can only hold if $\SSS(c_i)=\SSS(u)$ and
  $\SSS(c_j)=\SSS(v)$, which implies that $c_i=u$ and $c_j=v$. This shows
  that the endpoints of the selected edges $d_{1,2},\dotsc,d_{k-1,k}$
  are the vertices in $c_1,\dotsc,c_k$ and hence
  $G[\{c_1,\dotsc,c_k\}]$ is a $k$-clique of $G$.
\end{proof}


\begin{theorem}
\label{lem:binaryconstbdhard}
\textsc{ILP} is \NP{}\hy hard even if $\fr^C=1$.
\end{theorem}
\begin{proof}
  We show the result by a polynomial reduction from the
  \textsc{Subset Sum} problem, which is well-known to be weakly \NP{}\hy
  complete. Given a set $S:=\{s_1,\dotsc,s_n\}$ of integers and an
  integer $s$, the \textsc{Subset Sum} problem asks whether there is a
  subset $S' \subseteq S$ such that $\sum_{s \in S'}s'=s$.
  Let ${I:=(S,s)}$ with $S:=\{s_1,\dotsc,s_n\}$ be an instance of
  \textsc{Subset Sum}. We will construct an equivalent ILP instance $\III$ with
  $\fr^C(\III)=1$ in polynomial time as follows.
  The instance $\III$ has $n$ binary variables $x_1,\dotsc,x_n$ and apart from the
  domain constraints for these variables only one global constraint
  defined by $\sum_{1 \leq i \leq n}s_ix_i=s$. Because $\III$ has only one constraint, it holds
  that
  $\fr^C(\III)=1$ and moreover it is straightforward to verify that $\III$ is
  equivalent to $(S,s)$ (this has also for instance been shown
  in~\cite[Theorem 1]{JansenKratsch15esa}).
\end{proof}

%
At the end of this section we prove that ILP parameterized by coefficients and a constraint or variable backdoor does not admit a polynomial kernel, unless $\NP \subseteq \coNPpoly$.
We use polynomial parameter transformations from two problems which do not admit a polynomial kernel.

\pbDefPtask{{\sc Set Cover}}
{A universe $U$, a family ${\cF}$ of subsets of $U$, $k \in \nn$.}
{$|U|$.}
{Find a subfamily $\cF' \subseteq \cF$ such that $|\cF'| = k$ and $\cF'$ \emph{cover} $U$, i.e., $\bigcup_{F \in \cF'} F = U$.}

\pbDefPtask{{\sc Splitting Set}}
{A universe $U$, a family ${\cF}$ of subsets of~$U$.}
{$|U|$.}
{Find $X \subseteq U$ such that $X$ \emph{splits} each set $F \in \cF$, i.e., $F \cap X \neq \emptyset$ and $F \setminus X \neq \emptyset$.}

\begin{THE}[\cite{CyganFKLMPPS15}]
\label{thm:NoKernelProblems}
The problems {\sc Set Cover} and {\sc Splitting Set} do not admit a polynomial kernel, unless $\NP \subseteq \coNPpoly$.
\end{THE}

\begin{THE}
\textsc{ILP-feasibility} with 0 and 1 in the matrix, parameterized by $b^C_1(\III)$ and righthandside coefficients does not admit a polynomial kernel, unless $\NP \subseteq \coNPpoly$.
\end{THE}
\begin{proof}
Let $(U, \cF, k)$ is an instance of {\sc Set Cover}.
We formulate an ILP instance $\III$ with boolean variables ${x_F \in \{0,1\}}$ for each family $F \in \cF$.

\begin{align*}
\sum_{F \in \cF} x_F &\leq k & \\
\sum_{F \in \cF: u \in F} x_F &\geq 1 & \forall u \in U
\end{align*}

The meaning of the variable $x_F$ is we put $F$ in subfamily $\cF'$ if and only if $x_F = 1$.
Now, it is easy to see that $\III$ is feasible if and only if there is a subfamily $\cF' \subseteq \cF$ of size $k$ which covers $U$.
There are $|U| + 1$ constraints in the instance~$\III$, thus after removing them from the incidence graph $G_\III$ we get a graph without edges.
We use only $0$ and $1$ in the matrix.
The righthandside is bounded by $k$.
Note that if $k \geq |U|$, then the instance of {\sc Set Cover} is trivial.
Therefore, we can assume that $k \leq |U|$, i.e., it is bounded by the parameter.
The proof of theorem follows from Theorem~\ref{thm:NoKernelProblems}.
\end{proof}

\begin{THE}
\textsc{ILP-feasibility} with 0 and 1 in the matrix, parameterized by $b^V_1(\III)$ and righthandside coefficients does not admit a polynomial kernel, unless $\NP \subseteq \coNPpoly$.
\end{THE}
\begin{proof}
Let $(U, \cF)$ be an instance of {\sc Splitting Set}.
We formulate an ILP instance $\III$ with boolean variables ${y_u \in \{0,1\}}$ for each element $u \in U$.
\[
0 < \sum_{u \in F} y_u < |F| \hskip 1cm \forall F \in \cF
\]
The meaning of the variable $y_u$ is that we put the element~$u$ into $X$ if and only if $y_u = 1$.
Now, it is easy to see that the instance $\III$ is feasible if and only if there exists a set $X \subseteq U$ such that $X$ splits each set $F \in \cF$.
There are~$|U|$ variables in the instance $\III$, thus after removing them from the incidence graph $G_\III$ we get a graph without edges.
We use only $0$ and $1$ in the matrix and coefficients on the righthandside are bounded by $\max_{F \in \cF}|F| \leq |U|$.
The proof of theorem follows from Theorem~\ref{thm:NoKernelProblems}.
\end{proof}

\section{Concluding Notes}

In order to overcome the complexity barriers of ILP, a wide range of problems have been encoded in restricted variants of ILP such as $2$\hy stage stochastic ILP and $N$-fold ILP; examples for the former include a range of transportation and logistic problems~\cite{powell2003stochastic,hrabec2015}, while examples for the latter range from scheduling~\cite{KnopK16} to, e.g., computational social choice~\cite{KnopKM17}. Our framework 
provides a unified platform which generalizes $2$\hy stage stochastic ILP, $N$-fold ILP and also $4$-block $N$-fold ILP. Moreover, it represents a natural measure of the complexity of ILPs which can be applied to any ILP instance, including those which lie outside of the scope of all previously known algorithmic frameworks.
In fact, one may view our algorithmic results as ``algorithmic meta-theorems'' for ILP, where previously known algorithms for $2$\hy stage stochastic ILP, $N$-fold ILP and $4$-block $N$-fold ILP only represent a simple base case.

Our algorithms are complemented with matching lower bounds showing that the considered restrictions are, in fact, necessary in order to obtain fixed-parameter or \XP{} algorithms. The only remaining blank part in the presented complexity map is the question of whether mixed fracture backdoors admit a fixed-parameter algorithm in case of bounded coefficients; we consider this a major open problem in the area. A first step towards settling this question would be to resolve the fixed-parameter (in)tractability of $4$-block $N$-fold ILP; progress in this direction seems to require new techniques and insights~\cite{HemmeckeKW10}.

\section*{Acknowledgments}
Robert Ganian acknowledges support by the Austrian Science Fund (FWF, project P31336).

\bibliographystyle{named}
\bibliography{literature}

\end{document}